\documentclass{toc}

\usepackage{dsfont}
\usepackage{mathtools}
\usepackage{xspace}

\newcommand{\mc}[1]{\ensuremath{\mathcal{#1}}\xspace}
\newcommand{\mb}[1]{\ensuremath{\mathbb{#1}}\xspace}
\newcommand{\mv}[1]{\ensuremath{\mathbf{#1}}\xspace}
\newcommand{\tn}[1]{\ensuremath{\textnormal{#1}}\xspace}

\DeclareMathOperator{\Tr}{{\mathsf{Tr}}}
\newcommand{\Id}{\mathds{1}}

\newcommand{\Sym}{\mathbb S}

\newcommand{\eps}{\varepsilon} 
\renewcommand{\epsilon}{\varepsilon}

\newcommand{\pmset}[1]{\{-1,1\}^{#1}} %
\newcommand{\st}{:\,} %

\renewcommand{\Pr}{\mbox{\rm Pr}}

\newcommand{\Exp}{\mathbf{E}}
\newcommand{\parea}{\Lambda}
\DeclareMathOperator{\opt}{\mathsf{OPT}}
\DeclareMathOperator{\sdp}{\mathsf{SDP}}

\newcommand{\beq}{\begin{equation}}
\newcommand{\eeq}{\end{equation}}
\newcommand{\beqn}{\begin{equation*}}
\newcommand{\eeqn}{\end{equation*}}
\newcommand{\beqr}{\begin{eqnarray}}
\newcommand{\eeqr}{\end{eqnarray}}
\newcommand{\beqrn}{\begin{eqnarray*}}
\newcommand{\eeqrn}{\end{eqnarray*}}
\newcommand{\bmline}{\begin{multline}}
\newcommand{\emline}{\end{multline}}
\newcommand{\bmlinen}{\begin{multline*}}
\newcommand{\emlinen}{\end{multline*}}

\DeclarePairedDelimiter\inner{\langle}{\rangle}
\DeclarePairedDelimiter\ceil{\lceil}{\rceil}
\DeclarePairedDelimiter\length{\lVert}{\rVert}
\tocdetails{%
volume=13, number=15, year=2017, firstpage=1,
received={February 2, 2016},
revised={January 20, 2017},
published={December 2, 2017},
doi={10.4086/toc.2017.v013a015},
author={Jop Bri\"{e}t, Oded Regev, and Rishi Saket},
plaintextauthor={Jop Briet, Oded Regev, Rishi Saket},
title={Tight Hardness of the Non-Commutative Grothendieck Problem},
acmclassification={G.1.6},
amsclassification={68Q17, 15A60, 32A70, 03D15},
keywords={hardness of approximation, semidefinite programming, Grothendieck inequality}
}

\begin{document}

\begin{frontmatter}[classification=text]
\title{Tight Hardness of the\\ Non-Commutative Grothendieck Problem\titlefootnote{A conference version of this paper appeared in the 
Proceedings of the 56th Annual IEEE Symposium on
Foundations of Computer Science (FOCS
2015)~\cite{BrietRegevSaket:FOCS}.}}

\author[briet]{Jop Bri\"et\thanks{Supported by a Rubicon grant from the Netherlands Organisation for Scientific Research (NWO).}}
\author[regev]{Oded Regev\thanks{Supported by the Simons Collaboration on Algorithms and Geometry and by the National Science Foundation (NSF) under Grant No.~CCF-1320188. Any opinions, findings, and conclusions or recommendations expressed in this material are those of the authors and do not necessarily reflect the views of the NSF.}}
\author[saket]{Rishi Saket}

\begin{abstract}
We prove that for any~$\eps > 0$ it is \cclass{NP}-hard to approximate the non-commutative Grothendieck problem to within a factor~$1/2 + \eps$, which matches the approximation ratio of the algorithm of Naor, Regev, and Vidick (STOC'13).
Our proof uses an embedding of~$\ell_2$ into the space of matrices endowed with the trace norm with the property that the image of
standard basis vectors is longer than that of unit vectors with no large coordinates. 
We also observe that one can obtain a tight \cclass{NP}-hardness result for the \emph{commutative} Little Grothendieck problem;
previously, this was only known based on the 
Unique Games Conjecture    %
(Khot and Naor, Mathematika 2009).
\end{abstract}

%

\end{frontmatter}

\section{Introduction}

The  subject of this paper, the \emph{non-commutative Grothendieck problem}, has its roots in celebrated work of Grothendieck~\cite{Grothendieck:1953}, sometimes (jokingly?) referred to as ``\emph{Grothendieck's r\'esum\'e}.''
His paper  laid the foundation for the study of the geometry of tensor products of Banach spaces, though its significance  only became widely recognized after it was
revamped by Lindenstrauss and Pe\l czy\'nski~\cite{Lindenstrauss:1968}.
The main result of the paper, now known as \emph{Grothendieck's inequality}, shows a close relationship between the following two quantities.
For a complex~$d\!\times\!d$ matrix~$M$ let
\beq\label{eq:opt-groth}
\opt(M) = \sup_{\alpha_i,\beta_j}\Big|\sum_{i,j=1}^d M_{ij} \alpha_i \overline{\beta_j}\Big|\,,
\eeq
where the supremum goes over scalars on the complex unit circle, 
and let
\beq\label{eq:sdp-groth}
\sdp(M) =
\sup_{a_i,b_j}\Big|\sum_{i,j=1}^dM_{ij} \langle a_i, b_j\rangle\Big|\,, 
\eeq
where the supremum goes over vectors on a complex Euclidean unit sphere of any dimension.
Since the circle is the sphere in dimension one, we clearly have $\sdp(M)\ge \opt(M)$.
Grothendieck's inequality states that there exists a universal constant~$K_G^{\C}<\infty$ such that for 
any    %
positive integer~$d$ and any~$d\!\times\!d$ matrix~$M$, we also have
$
\sdp(M)
\leq 
K_G^{\C}\,\opt(M).
$
This result found an enormous number of applications both within and far beyond its original scope and we give some examples below (see~\cite{Khot:2012,Pisier:2012} for extensive surveys).
Despite this, finding the optimal value of~$K_G^\C$ is the only one of six problems posed in~\cite{Grothendieck:1953} that remains unsolved today; the current best upper and lower bounds are $1.4049$~\cite{Haagerup:1987} and $1.338$~\cite{Davie:1984}, respectively.
The situation is similar for the real variant of the problem, where all objects involved are over the real numbers.
The constant in that case is denoted~$K_G$ and is known to be between $1.6769$ and $1.7823$ (see~\cite{Braverman:2013}).

The non-commutative Grothendieck problem, 
to which we will refer  %
as the \emph{NCG}, is the optimization problem in which we are 
asked to maximize a given bilinear form over 
all pairs of  %
unitary matrices. More explicitly, we are given a four-dimensional array of complex numbers~$(T_{ijkl})_{i,j,k,l=1}^d$ and are asked to find or approximate the value
\beq\label{eq:opt-ncgt}
\opt(T) = \sup_{A,B} \Big|\sum_{i,j,k,l=1}^d T_{ijkl} A_{ij} \overline{B_{kl}}\Big|\,,
\eeq
where the supremum is over pairs of~$d\!\times\! d$ unitary matrices. (The word ``non-commutative'' simply refers to the fact that optimization is over matrices.)
It is not difficult to see that the (commutative) Grothendieck problem of computing $\opt(M)$ as in~\eqref{eq:opt-groth} is the special case where $T$ has $T_{iijj} = M_{ij}$ and zeros elsewhere.
Seen at first, the problem might seem overly abstract, but in fact, as 
we will illustrate below, it captures many natural questions as special cases. 
Grothendieck conjectured that his namesake inequality has an extension that relates~\eqref{eq:opt-ncgt} and the quantity
\beq\label{eq:sdp-ncgt}
\sdp(T) = \sup_{\vec{A}, \vec{B}} \Big|\sum_{i,j,k,l=1}^dT_{ijkl} \big\langle\vec{A}_{ij}, \vec{B}_{kl}\big\rangle\Big|\,,
\eeq
where~$\vec{A}, \vec{B}$ range over all~$d\!\times\!d$ matrices whose entries are complex \emph{vectors} of arbitrary dimension satisfying
a certain ``unitarity'' constraint.\footnote{Namely, we require that $\vec{A}^*\vec{A} = \Id$ and $\vec{A}\vec{A}^* = \Id$ and similarly for $\vec{B}$, 
where the multiplication of two vector-entried matrices is a scalar-valued matrix computed just like a normal matrix multiplication except
the scalar multiplication is replaced by an inner product, \eg, the $(i,j)$-coordinate of~$\vec{A}\vec{A}^*$ is given 
by~$\sum_k\inner{\vec{A}_{ik},\vec{A}_{jk}}$.}
Namely, he conjectured that there exists a universal constant~$K<\infty$ such that for every positive integer~$d$ and array~$T$ as above, we have~$\opt(T) \leq \sdp(T) \leq K \opt(T)$, where the first inequality follows immediately from the definition.
Over twenty-five years after being posed, the non-trivial content of Grothendieck's conjecture, $\sdp(T) \leq K \opt(T)$, was finally settled in the positive by Pisier~\cite{Pisier:1978}.
This result is now known as the non-commutative Grothendieck inequality.
In contrast with the commutative case, and somewhat surprisingly, the optimal value of~$K$ \emph{is} known:
Haagerup~\cite{Haagerup:1985} lowered Pisier's original estimate to~$K \leq 2$ and this was later shown to be sharp by Haagerup and Itoh~\cite{Haagerup:1995}.

\paragraph{Algorithmic applications.}
The importance of Grothendieck's inequality to computer science was pointed out by Alon and Naor~\cite{Alon:2006}, who placed it in the context of approximation algorithms for combinatorial optimization problems.
They observed that computing~$\sdp(M)$ is a semidefinite programming (SDP) problem that can be solved efficiently (to within arbitrary precision), and they translated an upper bound of about $1.78$ on~$K_G$ due to Krivine~\cite{Krivine:79a} to an efficient rounding scheme that turns SDP vectors into a feasible solution for the real Grothendieck problem~\eqref{eq:opt-groth} achieving value at least~$\sdp(M)/1.78$.
It is known that whatever the value of~$K_G$ is, there exists an efficient algorithm achieving value at least~$(1/K_G - \eps)\sdp(M)$ for any constant~$\eps>0$. (This was first shown in~\cite{Raghavendra:2009}, but can also be derived
from the results in~\cite{BrietFV10}
using a simple discretization argument combined with a brute-force search.)
The Grothendieck problem shows up in a number of different areas such as graph partition problems and computing the Cut-Norm of a matrix~\cite{Alon:2006}, in statistical physics where it gives ground state energies in the spin glass model~\cite{Kindler:2010}, and in quantum physics where it is related to Bell inequalities~\cite{Tsirelson:85b}.

In the same spirit, Naor, Regev, and Vidick~\cite{Naor:2013} recently translated the \emph{non-commutative} Grothendieck inequality into an efficient SDP-based approximation algorithm for the NCG problem~\eqref{eq:opt-ncgt} that achieves value at least~$\sdp(T)/2$.
They also considered the real variant and a Hermitian variant, for which they gave analogous algorithms achieving value at least~$\sdp(T)/2\sqrt{2}$.
This in turn implies efficient constant-factor approximation algorithms for a variety of problems, including the Procrustes problem and two robust versions of Principal Component Analysis~\cite{Naor:2013} and quantum XOR games~\cite{RegevV12a}. 
In a related result, Bandeira et al.~\cite{Bandeira:2013} 
considered a special case of the NCG (in fact a special case even 
of the Little NCG defined below) and showed how to obtain better approximation
factors for it. They also show that it is rich enough to capture
some applications such as the Procrustes problem and another natural problem called the Global Registration Problem.

\paragraph{Hardness of approximation.}
For simplicity we momentarily turn to the real setting, but similar results hold over the complex numbers.
The Grothendieck problem contains {\sc MaxCut} as a special case (in fact, it is a special case of the ``Little Grothendieck Problem,'' discussed below, in which~$M$ is the positive semidefinite Laplacian matrix of a graph~\cite{Alon:2006}).
It therefore follows from H\aa stad's inapproximability result~\cite{Hastad:2001} that it is \cclass{NP}-hard to approximate the value~\eqref{eq:opt-groth} to any factor larger than~$16/17 \approx .941$.
Based on the current best-known lower bound of about~$1.676$ on~$K_G$, Khot and O'Donnell~\cite{Khot:maxcutgain} proved that \eqref{eq:opt-groth} is Unique-Games-hard to approximate to within a factor larger than~$1/1.676\approx .597$.
Moreover, despite the fact that the exact value of~$K_G$ is still unknown, Raghavendra and Steurer~\cite{Raghavendra:2009} were able to improve this Unique Games hardness to~$1/K_G$.
(See for instance~\cite{Khot:CCC2010, Trevisan:2012} for background on the Unique Games conjecture.)

\paragraph{Our result.}
Whereas the hardness situation for the commutative version of Grothendieck's problem is reasonably well understood (apart from the yet-unknown exact value of $K_G$), no tight hardness result was previously known for the non-commutative version. In fact, we are not even aware of any hardness result that is better than what follows from the commutative case. 
Here we settle this question. 

\begin{theorem}\label{thm:ncgt-ughard}
For any constant~$\eps > 0$ it is \cclass{NP}-hard to approximate the optimum~\eqref{eq:opt-ncgt} of the non-commutative Grothendieck problem to within a factor greater than~$1/2+\eps$.
\end{theorem}

\paragraph{Little Grothendieck.}
In fact, we prove a stronger result than \expref{Theorem}{thm:ncgt-ughard} that concerns a special case of the NCG called the \emph{Little NCG}.
Let us start by describing the (real case of the) \emph{commutative} Little Grothendieck problem
(a.\,k.\,a.\ the positive-semidefinite Grothendieck problem).
A convenient way to phrase it is as asking for the operator norm of a linear map $\mathcal F:\R^n \to \ell_1^d$
(where $\ell_p^d$ denotes $\R^d$ endowed with the $\ell_p$ norm), 
defined as~$\length{\mathcal F} = \sup_a \length{{\mathcal F}(a)}_{\ell_1}$ where the vector~$a$ ranges over the~$n$-dimensional Euclidean unit ball.
It turns out that this is a special case of (the real version of) Equation~\eqref{eq:opt-groth}: for any ${\mathcal F}$ there exists a positive semidefinite~$d\!\times\!d$ matrix~$M$ such that~$\opt(M) = \length{\mathcal F}^2$; and vice versa, one can also map any such~$M$ into a corresponding operator~$\mathcal F$ (see, \eg,~\cite{Pisier:2012} or \expref{Section}{sec:little-big}).
We wish to highlight that for such instances, the constant~$K_G$ may be replaced by the smaller value~$\pi/2$~\cite{Rietz:1974} and that this value is known to be optimal~\cite{Grothendieck:1953}. 
Moreover, Nesterov made this algorithmic, namely, he showed an algorithm that approximates $\length{\mathcal F}$ as above
to within~$\sqrt{2/\pi}$~\cite{Nesterov:1998}.
Finally, Khot and Naor~\cite{Khot:2009}, as part of a more general result, showed that this is tight: the Unique-Games-hardness threshold for the Little Grothendieck problem is exactly $\sqrt{2/\pi}$.
As an aside, we note that other operator norms, in particular of operators in
$\R^n \to \ell_4^d$, have played an important role recently in theoretical computer science
(see, \eg, \cite{BarakBHKSZ12}).

The Little \emph{non-commutative} Grothendieck problem is formulated in terms of the (normalized) trace norm, also known as the Schatten-1 norm, which for a $d\!\times\!d$ matrix~$A$ is given by $\length{A}_{S_1} = d^{-1}\Tr \sqrt{A^*A}$. 
In other words, $\length{A}_{S_1}$ is the average of the singular values of $A$.
The space of matrices endowed with this norm is denoted by~$S_1$ and by~$S_1^d$ if we restrict to~$d\!\times\!d$ matrices.
The problem then asks for the operator norm of a linear map~$\mathcal F:\C^n\to S_1^d$.
This problem is a special case of the NCG where $\opt(T) = \length{\mathcal F}^2$ (see \expref{Section}{sec:little-big}).
In particular, it follows from~\cite{Haagerup:1985, Naor:2013} that there is an efficient SDP-based $1/\sqrt{2}$-approximation algorithm for the Little NCG\@.
Our stronger result alluded to above shows tight hardness for the Little NCG, which directly implies \expref{Theorem}{thm:ncgt-ughard}.

\begin{theorem}\label{thm:lncg-ughard}
  \begin{sloppypar}
    For any constant~$\eps>0$ it is \cclass{NP}-hard to approximate the Little non-commutative Grothendieck problem to within a factor greater than~$1/\sqrt{2} + \eps$.
    \end{sloppypar}
\end{theorem}

While this result applies to the complex case, an easy transformation shows that it directly implies the same result for the real and Hermitian cases introduced in~\cite{Naor:2013} (see \expref{Section}{sec:NC-RH}).
Finally, as we show in the ``warm-up'' section of this paper (\expref{Section}{sec:comm}), we also get a tight \cclass{NP}-hardness result for the \emph{commutative} Little Grothendieck problem, strengthening the unique-games-based result of~\cite{Khot:2009}.

\begin{theorem}\label{thm:comm-hard}
For any constant~$\eps>0$ it is \cclass{NP}-hard to approximate the real Little commutative Grothendieck problem to within a factor greater than~$\sqrt{2/\pi} + \eps$.
Similarly, the complex case is \cclass{NP}-hard to approximate to within a factor greater than~$\sqrt{\pi/4} + \eps$.
\end{theorem}

\paragraph{Techniques.}
Nearly all recent work on hardness of approximation, including for commutative Grothen\-dieck problems~\cite{Raghavendra:2009, Khot:2009}, 
uses the machinery of Fourier analysis over the hypercube, influences, or 
the
\emph{majority is stablest} %
theorem~\cite{ODonnell:2014}. %
Our attempts to apply these techniques here failed. 
Instead, we use a more direct approach similar to that taken in~\cite{GRSW} and avoid the use of the hypercube altogether. 
The role of dictator functions is played in our proof simply by the standard basis vectors of $\C^n$. The dictatorship test, which is our main technical contribution, comes in the form of a linear operator ${\mathcal F}:\C^n \to S_1^d$ with the following notable property: it maps the $n$ standard basis vectors to matrices with trace norm~1, and it maps any unit vector with no large coordinate to a matrix with trace norm close to~$1/\sqrt{2}$. Roughly speaking, one can think of $\mathcal F$ as identifying an interesting subspace of $S_1^d$ (namely, the image of $\mathcal F$) in which the unit ball looks somewhat like the intersection of the Euclidean ball 
of radius $\sqrt{2}$ with the $\ell_\infty$ ball $[-1,1]^n$ (since with such a unit ball, the norm of a vector $a$ is given by $\max\{\|a\|_2/\sqrt{2},\|a\|_\infty\}$).

A first attempt to construct an operator $\mathcal F$ as above might be to map each standard basis vector to a random unitary matrix. This, however, leads to a very poor map---while standard basis vectors are mapped to matrices of trace norm 1, vectors with no large coordinates are mapped to matrices of trace norm close to $8/(3\pi) \approx 0.848$ by Wigner's semicircle law.
Another natural approach is to look at the construction by
Haagerup and Itoh~\cite{Haagerup:1995} (see also~\cite[Section 11]{Pisier:2012} for a self-contained description) which shows the factor-2 lower bound 
in the non-commutative Grothendieck inequality, \ie, the tight integrality gap of the SDP~\eqref{eq:sdp-ncgt}.
Their construction relies on the so-called CAR algebra (after canonical anticommutation relations) and provides
an \emph{isometric} mapping from~$\C^n$ to~$S_1$, \ie, all unit vectors are mapped to matrices of trace norm $1$. 
Directly modifying this construction (akin to how, \eg, Khot et al.~\cite{KKMO} obtained tight hardness of {\sc MaxCut}
by restricting the tight integrality gap instances by Feige and Schechtman~\cite{FeigeS02} from the sphere to the hypercube) does not 
seem to work. 
Instead, our construction of~$\mathcal F$ relies on a different (yet related) algebra known as the Clifford algebra.
The Clifford algebra was used before in a celebrated result by Tsirelson~\cite{Tsirelson:85b} (to show that Grothendieck's inequality can be interpreted as a statement about XOR games with entanglement).
His result crucially relies on the fact that the Clifford algebra gives an \emph{isometric} mapping from~$\R^n$ to~$S_1$.
Notice that this is again an isometric embedding, but now only over the reals. 
Our main observation here (\expref{Lemma}{lem:Tmap}) is that the same mapping, when extended to $\C^n$, 
exhibits intriguing cancellations when the phases in the input vectors are not aligned, and this
leads to the construction of~$\mathcal F$ (\expref{Lemma}{lem:ourmap}).
Even though the proof of this fact is simple, we find it surprising; 
we are not aware of any previous application of such ``complex extensions'' of Clifford algebras.

\paragraph{Open questions.}
For the real and Hermitian cases there is a gap of~$\sqrt{2}$ between the guarantee of the~\cite{Naor:2013} algorithms and our hardness result.
It also would be interesting to explore whether hardness of approximation results can be derived to some of the applications of the NCG, including the Procrustes problem and robust Principle Component Analysis. We believe that our embedding would be useful there too.

\paragraph{Outline.}
The rest of the paper is organized as follows.
In \expref{Section}{sec:prelims}, we set some notational conventions, gather basic preliminary facts about relevant Banach spaces, and give a detailed formulation of the Smooth Label Cover problem.
In \expref{Section}{sec:FXsec}, we prove hardness of approximation for the problem of computing the norm of a general class of Banach-space-valued functions, closely following~\cite{GRSW}.
In \expref{Section}{sec:comm}, as a ``warm up,'' we prove 
\expref{Theorem}{thm:comm-hard} using the generic result of \expref{Section}{sec:FXsec} and straightforward applications of real and complex versions of the Berry-Ess\'een Theorem. 
\expref{Section}{sec:NC} contains our main technical contribution, which we use there to finish the proof of our main result (\expref{Theorem}{thm:lncg-ughard}).

\paragraph{Acknowledgements.} We thank Steve Heilman and Thomas Vidick for early discussions. We also thank Gilles Pisier and anonymous referees for helpful comments on an earlier version of this manuscript.

\section{Preliminaries}
\label{sec:prelims}

\paragraph{Notation and relevant Banach spaces.}
For a positive integer~$n$ we denote $[n] = \{1,\dots,n\}$.
For a graph~$G$ and vertices~$v,w\in V(G)$ we write~$v\sim w$ to denote that~$v$ and~$w$ are adjacent.
We write~$\Pr_{e\sim v}[\cdot]$ for the probability
with respect to a uniformly distributed random edge 
incident with~$v$.  %
For a finite set~$U$ we denote by~$\Exp_{u\in U}[\cdot]$ the expectation with respect to the uniform distribution over~$U$.
For a complex number~$c\in \C$, we denote its real and imaginary parts by~$\Re(c)$ and $\Im(c)$, respectively.
All Banach spaces are assumed to be finite-dimensional (so we can equivalently talk about normed spaces).
Recall that for Banach spaces~$X,Y$ the operator norm of a linear operator~$\mathcal F:X\to Y$ is given by
\beqn
\length{\mathcal F} = \sup_{x\in X\st \length{x}_X \leq 1}\length{\mathcal F(x)}_Y\,.
\eeqn
For a real number~$p\geq 1$, the~$p$-norm of a vector~$a\in\C^n$ is given by
\beqn
\length{a}_{\ell_p} = \left(\sum_{i=1}^n|a_i|^p\right)^{1/p}\,. 
\eeqn
As usual we implicitly endow~$\C^n$ with the Euclidean norm~$\length{a}_{\ell_2}$.
For a finite set~$U$ endowed with the uniform probability measure we denote by $L_p(U)$ the space of  functions $f:U\to\C$ with the norm
\beqn
\length{f}_{L_p(U)} = \left(\Exp_{u\in U}\big[|f(u)|^p\big]\right)^{1/p}\,.
\eeqn
More generally, for a Banach space~$X$ we denote by~$L_p(U, X)$ the space of functions $f:U\to X$ with the norm
\beqn
\length{f}_{L_p(U,X)} = \left(\Exp_{u\in U}\Big[ \length{f(u)}_X^p \Big]\right)^{1/p}\,.
\eeqn
We will write~$L_p(X)$ if~$U$ is not explicitly given and~$\length{f}_{L_p}$ instead of $\length{f}_{L_p(U,X)}$ when there is no danger of ambiguity.
Note that~$L_2(U, \C^n)$ is a Hilbert space.

\paragraph{Smooth Label Cover.}
An instance of \textnormal{Smooth Label Cover} is given by a quadruple~$(G, [n], [k], \Sigma)$ that consists of a regular connected (undirected) graph~$G = (V, E)$, a label set~$[n]$ for some positive integer~$n$, and a collection
\[
  \Sigma = \big((\pi_{ev}, \pi_{ew})\st e=(v,w)\in E\big)
\]
of pairs of maps both from~$[n]$ to~$[k]$ associated with the endpoints of the edges in~$E$.
Given an \emph{assignment} $A:V\to[n]$, we say that an edge~$e=(v,w)\in E$ is \emph{satisfied} if $\pi_{ev}\big(A(v)\big) = \pi_{ew}\big(A(w)\big)$.

The following hardness result for Smooth Label Cover, 
given in~\cite{GRSW},\footnote{For convenience, we make implicit some of
the parameters in the statement of the theorem.}
is a slight variant of the original
construction due to Khot~\cite{Khot02-color}. 
The theorem also describes the various structural properties, 
including smoothness, that are satisfied by the hard instances.

\begin{theorem}\label{thm:sml}
For any positive real numbers $\zeta, \gamma$ there exist positive integers ${n = n(\zeta, \gamma)}$, $k = k(\zeta,\gamma)$, and~$t = t(\zeta)$, and a
family of Smooth Label Cover instances $(G, [n], [k], \Sigma)$ as above 
such that %
\begin{itemize}
\item (Hardness)\ \ %
        It is \cclass{NP}-hard to distinguish between the following two cases:
	\begin{itemize} 
	\item(YES Case)\ \ %
          There is an assignment that satisfies all edges. 

	\item(NO Case)\ \ %
         Every assignment satisfies less than a $\zeta$-fraction of the edges.
	\end{itemize} 
\item (Structural properties)\ \ %
	\begin{itemize} 

	\item (Smoothness)\ \ %
For every vertex $v \in V$ and distinct~$i,j\in [n]$, we have 
\beq\label{eq:smoothness}
\Pr_{e\sim v} \left[\pi_{ev}(i) = \pi_{ev}(j)\right] \leq \gamma\,.
\eeq

	\item For every vertex $v\in V$, edge $e\in E$ incident on $v$, and $i \in [k]$, we have $|\pi_{ev}^{-1}(i)|\leq t $;
that is, at most $t$ elements in $[n]$ are mapped to
the same element in $[k]$.

	\item (Weak Expansion)\ \ %
For any $\delta > 0$ and vertex subset $V'\subseteq V$ such that
$|V'| = \delta\cdot |V|$, the number of edges between the vertices in $V'$
 is at least $(\delta^2/2) |E|$.  
	\end{itemize} 
\end{itemize}
\end{theorem}

\section{Hardness for general Banach-space valued operators}
\label{sec:FXsec}

The following proposition shows hardness of approximation for the problem of computing the norm of a linear map from~$\C^n$ to any Banach space that allows for a ``dictatorship test,'' namely, a linear function that maps the standard basis vectors to long vectors, and maps ``spread'' unit vectors to short vectors.
As stated, the proposition assumes the underlying field to be~$\C$; 
we note that the proposition holds with exactly the same proof also in the case of the real field $\R$.

\begin{theorem}\label{thm-NP-hard}
	Let $(X_n)_{n\in \mb{N}}$ be a family of
	finite-dimensional Banach spaces, and $\eta$ and $\tau$ be
	positive numbers such that $\eta > \tau$. Suppose that for
	each positive integer $n$ there exists a linear operator ${f :
	\mb{C}^n \to X_n}$ with the following properties:
	\begin{itemize}
		\item For any vector $a \in \mb{C}^n$, we have
			$\|f(a)\|_{X_n}\leq \|a\|_{\ell_2}$.
		\item For each standard basis vector $e_i$, we have
			$\|f(e_i)\|_{X_n}\geq \eta$.
		\item For any $\eps > 0$, there is a $\delta =
			\delta(\eps) > 0$ such that $\|f(a)\|_{X_n} >
			(\tau + \eps)\|a\|_{\ell_2}$ implies
			$\|a\|_{\ell_4} > \delta\|a\|_{\ell_2}$.
	\end{itemize}
	Then, for any $\eps' > 0$ there exists a positive integer $n$
	such that it is \cclass{NP}-hard to approximate the norm of an
	explicitly given linear operator $\mc{F}: L_2 \to
	L_1(X_n)$ to within a factor greater than $(\tau/\eta) +
	\eps'$. 
\end{theorem}

\subsection{The hardness reduction}

To set up the reduction, we begin by defining a linear operator~$\mathcal F = \mathcal F_{\zeta,\gamma}$ for any choice of positive real numbers~$\zeta,\gamma$.
Afterwards we show that there is a choice of these parameters giving the desired result.
For positive real numbers $\zeta, \gamma$, let~$n$, $k$, and~$t$ be positive integers (depending on~$\zeta, \gamma)$ and~$(G, [n], [k], \Sigma)$ a Smooth Label Cover instance as in \expref{Theorem}{thm:sml}, where~$G = (V,E)$ is a regular graph.
Note that~$\zeta$ controls the ``satisfiability'' of the instance in the NO case, that~$\gamma$ controls the ``smoothness,'' and that~$t$ depends on~$\zeta$ only.
Endow the vertex set~$V$ with the uniform probability measure.
To define~$\mathcal F$ we consider a special linear subspace~$\mathcal H$ of the Hilbert space~$L_2(V, \C^n)$.
It will be helpful to view a vector~${\bf a}\in L_2(V, \C^n)$ as an assignment ${\bf a} = (a_v)_{v\in V}$ of vectors~$a_v\in \C^n$ to~$V$.
Let~$\mathcal H\subseteq L_2(V, \C^n)$ be the subspace of vectors~${\bf a} = (a_v)_{v\in V}$ that satisfy for every $e = (v,w) \in E$ and $j \in [k]$ the
homogeneous linear constraint
\begin{equation} 
	 \sum_{i \in \pi_{ev}^{-1}(j)}a_{v}(i) = 
	\sum_{i \in \pi_{ew}^{-1}(j)}a_{w}(i)\,, \label{eqn-Q-def}
\end{equation}
where~$a_v(i)$ denotes the~$i$th coordinate of the vector~$a_v$.
Notice that if an assignment~$A:V\to[n]$ satisfies the edge~$e = (v,w)$, then the standard basis vectors~$a_v = e_{A(v)}$ and $a_w = e_{A(w)}$ satisfy~\eqref{eqn-Q-def}; 
indeed, if~$\pi_{ev}\big(A(v)\big) = \pi_{ew}\big(A(w)\big) = j'$ then both sides of~\eqref{eqn-Q-def} equal~1 if~$j = j'$ and equal zero otherwise.

Now let $\eta$, $\tau$ and $f$ be as in \expref{Theorem}{thm-NP-hard}.
We associate with the Smooth Label Cover instance $(G, [n], [k], \Sigma)$ from above the linear operator
$\mc{F} : \mc{H} \to L_1(V, X_n)$ given by
\begin{equation}
	\left(\mc{F}(\mv{a})\right)(v) = f(a_v)\,. \label{eqn-F-def}
\end{equation}
The operator~$\mathcal F$ thus maps a $\C^n$-valued assignment ${\bf a} = (a_v)_{v\in V}$ satisfying~\eqref{eqn-Q-def} to an~$X_n$-valued assignment given by~$f(a_v)$ for each~$v\in V$.
\expref{Theorem}{thm-NP-hard} follows from the following two lemmas, which we prove in Sections~\ref{sec:completeness} and~\ref{sec:soundness}, respectively.

\begin{lemma}[Completeness]\label{lem:completeness}
Suppose that there exists an assignment $A: V\to [n]$ that satisfies all the edges in~$E$.
Then, $\length{\mathcal F} \geq \eta$.
\end{lemma}

\begin{lemma}[Soundness]\label{lem:soundness}
For any~$\eps >0$ there exists a choice of~$\zeta,\gamma >0$ such that if ${\length{\mathcal F} > \tau + 4\eps}$ then there exists an assignment that satisfies at least a~$\zeta$-fraction of the edges of~$G$. 
\end{lemma}

\begin{proof}[Proof of \expref{Theorem}{thm-NP-hard}]
Let $\eps>0$ be arbitrary, let $\zeta,\gamma$ be as in \expref{Lemma}{lem:soundness}, and let $n=n(\zeta,\gamma)$ and~$k = k(\zeta,\gamma)$ be as in \expref{Theorem}{thm:sml}.  
We use the reduction described above, which maps a Smooth Label Cover instance $(G, [n], [k], \Sigma)$
to the linear operator $\mathcal F:L_2\to L_1(X_n)$ specified in~\eqref{eqn-F-def}.
By \expref{Lemma}{lem:completeness}, YES instances are mapped to ${\mathcal F}$ satisfying $\length{\mathcal F}\geq \eta$, whereas by \expref{Lemma}{lem:soundness}, NO instances are mapped
to $\mathcal F$ satisfying $\length{\mathcal F} < \tau + 4\eps$.
We therefore obtain hardness of approximation to within a factor $(\tau + 4\eps)/\eta$. Since $\eps$ is arbitrary, we are done. 
\end{proof}

\subsection{Completeness}
\label{sec:completeness}

Here we prove \expref{Lemma}{lem:completeness}.

\begin{proof}[Proof of \expref{Lemma}{lem:completeness}]
Let $A : V \to [n]$ be an assignment that satisfies all the edges.
Consider the vector $\mv{a} \in L_2(V, \C^n)$ where $a_v =
e_{A(v)}$ and notice that ${\|\mv{a}\|_{L_2} = 1}$.
Since~$A$ satisfies all edges, $\mv{a}$ satisfies the constraint~\eqref{eqn-Q-def} for every $e \in E$ 
and $j \in [n]$, and thus $\mv{a}$ lies in the domain~$\mc{H}$ of~$\mathcal F$.
Moreover, by the second property of~$f$ given in \expref{Theorem}{thm-NP-hard},
\beqn
\|\mc{F}(\mv{a})\|_{L_1(V,X_n)}  = \Exp_{v\in V}\big[\length{f(a_v)}_{X_n}\big] \geq \eta.
\eeqn
Hence, $\|\mc{F}\| \geq \eta$.
\end{proof}

\subsection{Soundness}
\label{sec:soundness}

Here we prove \expref{Lemma}{lem:soundness}
and show that among the family of operators~$\mathcal F = \mathcal F_{\zeta,\gamma}$ as in~\eqref{eqn-F-def}, for any~$\eps>0$ there is a choice of~$\zeta,\gamma>0$ such that if $\length{\mathcal F} > \tau + 4\eps$, then there exists an assignment satisfying a~$\zeta$-fraction of the edges in the Smooth Label Cover instance associated with~$\mathcal F$.
To begin, assume that
$	\|\mc{F}\| > \tau + 4\eps$
for some ${\eps> 0}$.
Let $\mv{b} \in
\mc{H}$ be a vector such that $\|\mv{b}\|_{L_2} = 1$ and
\begin{equation}
	\Exp_{v\in V}\big[\length{f(b_v)}\big]
	=
	\length{\mathcal F({\bf b})}_{L_1}
	\geq
	\tau + 4\eps\,.\label{eq:Ftaueps}
\end{equation}
The weak expansion property in \expref{Theorem}{thm:sml} implies that it suffices to find a ``good'' assignment for a large subset of the vertices, as any large set of vertices will induce a large set of edges.
For~$\delta = \delta(\eps)$ as in \expref{Theorem}{thm-NP-hard}, we will consider set of vertices
\begin{equation}
	V_0 = \{v \in V\, \mid\, \|b_v\|_{\ell_4} > \delta\eps\:\:\tn{
	and }\:\:\|b_v\|_{\ell_2} \leq 1/\eps\}\,.\label{eqn-vzerodef}
\end{equation}

The following lemma shows that $V_0$ contains a significant fraction of vertices.

\begin{lemma}\label{lem-Vzero}
	For~$V_0\subseteq V$ defined as in~\eqref{eqn-vzerodef}, we have $|V_0| \geq \eps^2|V|$.
\end{lemma}
\begin{proof}

Define the sets
\begin{align*}
	V_1 & =  \{v \in V\,\mid\, \|b_v\|_{\ell_4} \leq
\delta\eps\:\:\tn{ and }\:\: \|b_v\|_{\ell_2} < \eps\}\,, \\
	V_2 & =  \{v \in V\,\mid\, \|b_v\|_{\ell_4} \leq
\delta\eps\:\:\tn{ and }\:\: \|b_v\|_{\ell_2} \geq \eps\}\,, \\
	V_3 & =  \{v \in V\,\mid\, \|b_v\|_{\ell_2} > 1/\eps\}\,.
\end{align*}
	From~\eqref{eq:Ftaueps}, we have
	\begin{equation}
		\sum_{v\in V_0}\|f(b_v)\|_{X_n} + 
		\sum_{v\in V_1}\|f(b_v)\|_{X_n} +
		\sum_{v\in V_2}\|f(b_v)\|_{X_n} +
		\sum_{v\in V_3}\|f(b_v)\|_{X_n}
		\geq (\tau + 4\eps)|V|\,. \label{eqn-3termsum}
	\end{equation}
	We bound the four sums on the left-hand side of~\eqref{eqn-3termsum} individually.
	Since (by the first item in \expref{Theorem}{thm-NP-hard}) we have $\|f(b_v)\|_{X_n} \leq \|b_v\|_{\ell_2}$, and since $\|b_v\|_{\ell_2} \leq 1/\eps$ for every $v \in V_0$,
	the first sum in~\eqref{eqn-3termsum} can be bounded by
	\begin{equation}
		\sum_{v\in V_0}\|f(b_v)\|_{X_n} \leq |V_0|/\eps\,.
	\end{equation}
	Similarly, using the definition of $V_1$ the second
	sum in~\eqref{eqn-3termsum} is at most~$\eps|V|$. 
	Next, from the third property of $f$ in \expref{Theorem}{thm-NP-hard}, for each $v \in V_2$, we have $\|f(b_v)\|_{X_n} \leq (\tau + \eps)\|b_v\|_{\ell_2}$. 
	Therefore, the third sum in~\eqref{eqn-3termsum} is bounded as
	\begin{align}
		\sum_{v\in V_2}\|f(b_v)\|_{X_n} &\leq 
		(\tau + \eps)\sum_{v\in V_2}\|b_v\|_{\ell_2}
		\nonumber &&\\
		&\leq (\tau + \eps){|V_2|}^{\frac{1}{2}}
		\Big(\sum_{v\in
		V_2}\|b_v\|_{\ell_2}^2\Big)^{\frac{1}{2}} &&
		\tn{(by Cauchy-Schwartz)} \nonumber \\
		&\leq (\tau + \eps){|V|}^{\frac{1}{2}}
		\Big(\sum_{v\in
		V}\|b_v\|_{\ell_2}^2\Big)^{\frac{1}{2}} \nonumber &&\\
		&=  (\tau + \eps)|V|\,, &&
	\end{align}
	where the last inequality uses $\|\mv{b}\|_{L_2} = 1$.
	Finally, the fourth sum in~\eqref{eqn-3termsum} is bounded by
	\begin{align}
		\sum_{v\in V_3}\|f(b_v)\|_{X_n}&\leq
		\sum_{v\in V_3} \|b_v\|_{\ell_2} &&
		\tn{(by
		the property of }f\tn{)} \nonumber \\
		&< \sum_{v\in V_3} \eps\|b_v\|^2_{\ell_2}
		 \nonumber && \\
		&\leq \eps\sum_{v \in V} \|b_v\|_{\ell_2}^2  \nonumber &&\\
		&= \eps |V|\|\mv{b}\|_{L_2}^2 = \eps|V|\,. &&
	\end{align}
	Combining the above with Equation \eqref{eqn-3termsum} yields,
	$|V_0|/\eps \geq \eps|V|$,
	which proves \expref{Lemma}{lem-Vzero}.
\end{proof}
\expref{Lemma}{lem-Vzero} and the weak expansion property implies that the set~$E(V_0)$ of edges induced by~$V_0$ has cardinality
\begin{equation}
	|E(V_0)| \geq (\eps^4/2)|E|\,. \label{eqn-dense-edges}
\end{equation}

We set out to show that there exists an assignment to the vertices in
$V_0$ that satisfies a significant fraction of edges in $E(V_0)$.
Roughly speaking, we do this by randomly assigning each ${v\in V_0}$ one of the coordinates of the vector~$b_v$ at which it has large magnitude. (Assigning the largest coordinate may not work.)
The following simple proposition shows that those vectors indeed have large coordinates.

\begin{proposition}\label{prop-max-element}
	Let~$\beta = \beta(\eps)$ be given by~$\beta = \delta^2\eps^3$.
	Then, for each $v \in V_0$, we have $\length{b_v}_{\ell_\infty} \geq \beta$.
\end{proposition}

\begin{proof}
	For every~$v\in V_0$, we have
	\begin{equation}
		\delta^4\eps^4 \leq \|b_v\|_{\ell_4}^4 \leq \length{b_v}_{\ell_\infty}^2
		\|b_v\|_{\ell_2}^2 \leq \length{b_v}_{\ell_\infty}^2/\eps^2\,,
	\end{equation}
	giving the claim.
\end{proof}

For the to-be-determined value of~$\zeta$ let~$t = t(\zeta)$ be as in \expref{Theorem}{thm:sml} and
for each $v \in V_0$ let

\beqn
	A^v_1 = \Big\{i \in [n]\, \Big|\, |b_v(i)| \geq \frac{\beta}{4}\Big\}
	\quad\quad
	\text{and}
	\quad\quad
	A^v_2 = \Big\{i \in [n]\, \Big|\, |b_v(i)| \geq \frac{\beta}{4t}\Big\}\,.
\eeqn

By \expref{Proposition}{prop-max-element} these sets are nonempty and clearly  $A^v_1\subseteq A^v_2$.
Moreover, since $\|b_v\|_{\ell_2} \leq 1/\eps$, we have,
\begin{equation}
	|A^v_1| \leq \frac{16}{\eps^2\beta^2}
	\quad\quad 
	\tn{and}
	\quad\quad 
	|A^v_2| \leq \frac{16t^2}{\eps^2\beta^2}\,. \label{eqn-sizes}
\end{equation}

Now consider a random assignment~$A:V_0\to[n]$ that independently assigns each vertex $v \in V_0$ a uniformly random label from $A^v_1$ and assigns the remaining vertices in~$V$ some fixed arbitrary label. The following
lemma shows that on average, this assignment satisfies a significant fraction of
edges.

\begin{lemma}\label{lem-good-labeling}
	There exists a~$\gamma > 0$
	depending only on $\eps$
	and $\zeta$ such that for some absolute constant~$c>0$
	the expected fraction of edges in $E$ satisfied by the random assignment~$A$ given above is at least~$c\eps^8\beta^4$.
\end{lemma}

Setting $\gamma$ appropriately as in the above lemma and
$\zeta = c\eps^8\delta^4$ then gives \expref{Lemma}{lem:soundness}; 
indeed, notice that then~$\zeta$, and therefore also~$\gamma$, depend on~$\eps$ alone.

The remainder of this section is devoted to the proof of \expref{Lemma}{lem-good-labeling}.
Let $E' \subseteq E(V_0)$ be the subset of edges $e = (v,w)$ whose 
projections $\pi_{ev}$ and $\pi_{ew}$ are injective on 
the subsets $A^v_2$ and
$A^w_2$ respectively. Formally,
\begin{equation}
	E' = \big\{e = (v,w) \in E(V_0)\, \big|\, |\pi_{ev}(A^v_2)| =
	|A^v_2|,\tn{ and }|\pi_{ew}(A^w_2)| = |A^w_2|\big\}\,.
\end{equation}
We set the parameter $\gamma$ according to the
following proposition which shows a lower bound on $|E'|$ using the smoothness
property.
Recall that~$t$ is a function of~$\zeta$ only.

\begin{proposition}\label{prop:Eprime}
	There exists an absolute constant~$c'>0$ such that for any $\gamma \leq c'\eps^8\beta^4/t^4$, the set~$E'$ has cardinality $|E'| \geq
	(\eps^4/4)|E|$\,.
\end{proposition}

\begin{proof}
	Consider any vertex $v \in V_0$. By the smoothness property
	of \expref{Theorem}{thm:sml} and a union bound over all distinct pairs~$i,j\in A_v^2$, the fraction of
	edges $e \in E$ incident on $v$ that \emph{do not} satisfy
	\begin{equation}
		|\pi_{ev}(A^v_2)| = |A^v_2| \label{eqn-biject}
	\end{equation}
	is at most 
	\[ 
	\frac{\gamma\, |A^v_2|^2}{2} \leq
	\frac{1}{2}\left(\frac{\eps^8\beta^4}{2^{10}\cdot 
	t^4}\right)\left(\frac{16^2\cdot t^4}{\eps^4\beta^4}\right) =
	\frac{\eps^4}{8}\,,
	\]
	via an appropriate setting of~$c'$. Therefore,
	the number of edges in $E$ that are incident on some
	$v \in V_0$ and do not satisfy~\eqref{eqn-biject} is
	at most
	\beqn
	\sum_{v\in V_0}\frac{\eps^4}{8}\deg(v) 
	\leq
	\frac{\eps^4}{8}\sum_{v\in V}\deg(v)\\
	\leq \frac{\eps^4}{4}|E|\,.
	\eeqn
	Thus, 
	$$|E'| \geq |E(V_0)| - (\eps^4/4)|E| \geq (\eps^4/4)|E|\,,$$
	by Equation \eqref{eqn-dense-edges}.
\end{proof}
The following proposition shows that for an edge $e = (v,w) \in E'$, 
the label sets $A^v_1$ and $A^w_1$ intersect under
projections given by $e$.

\begin{proposition}\label{prop-intersect}
	For every edge $e =(v,w) \in E'$, we have $\pi_{ev}(A^v_1)\cap\pi_{ew}(A^w_1)
	\neq \emptyset$.
\end{proposition}

\begin{proof}
	From \expref{Proposition}{prop-max-element}, let $i^* \in [n]$ be such that
	$|b_v(i^*)| \geq \beta$. Note that~$i^*\in A_1^v$. Let ${j^* = \pi_{ev}(i^*)}$.
	Clearly it suffices to show that there exists an~$i'\in A_1^w$ such that~${\pi_{ew}(i') = j^*}$, as this implies that~$j^* \in \pi_{ev}(A^v_1)\cap\pi_{ew}(A^w_1)$.
	
	Recall that since $\mv{b}\in \mc{H}$, the vector
	$\mv{b}$ satisfies the constraint~\eqref{eqn-Q-def}, in particular,
	\begin{equation}
		\bigg| \sum_{i \in \pi_{ev}^{-1}(j^*)}b_v(i) \bigg| =  
		\bigg|\sum_{i\in \pi_{ew}^{-1}(j^*)}b_w(i)\bigg|\,.\label{eq:binH}
	\end{equation}
	We show that because $i^* \in A_1^v$, the left-hand side must be large. 
	Therefore the right hand side is also large, from which we conclude that there must exist a coordinate 
	$i'\in \pi_{ew}^{-1}(j^*)$ such that $|b_w(i')|$ is large, and so $i' \in A_1^w$.
	
	Recall from the second structural property in \expref{Theorem}{thm:sml} that~$|\pi_{ev}^{-1}(j^*)| \leq t$.
	Moreover, since~$\pi_{ev}$ acts injectively on the set~$A_2^v$ and since~$i^*\in A_2^v$, no index~$i\ne i^*$ such that~$\pi_{ev}(i)= \pi_{ev}(i^*)$ can belong to~$A_2^v$. 
	Hence, by the triangle inequality, the left-hand side of~\eqref{eq:binH} is at least
	\begin{align}
		|b_v(i^*)| - \sum_{\substack{i \in
		\pi_{ev}^{-1}(j^*) \\ i\neq i^*}}|b_v(i)| 
		\geq \beta - t\cdot\left(\frac{\beta}{4t}\right)
		 = \frac{3\beta}{4}\,. \label{eqn-vsidesum}
	\end{align}
	
Combining~\eqref{eq:binH},~\eqref{eqn-vsidesum}, and the triangle inequality lets us bound the right-hand side of~\eqref{eq:binH} by
	\begin{align}
	\frac{3\beta}{4}
	&\leq
	\left|\sum_{i\in \pi_{ew}^{-1}(j^*)}b_w(i)\right| \nonumber\\
	&\leq
	\sum_{i\in \pi_{ew}^{-1}(j^*)\cap A_2^w}|b_w(i)|
	+
	\sum_{i\in \pi_{ew}^{-1}(j^*)\smallsetminus A_2^w}|b_w(i)| \nonumber\\
	&\leq
	\sum_{i\in \pi_{ew}^{-1}(j^*)\cap A_2^w}|b_w(i)|
	+
	t\,\frac{\beta}{4t}\,,\label{eq:RHS}
	\end{align}
	where the last inequality uses the same facts as above.
Since~$\pi_{ew}$ acts injectively on~$A_2^w$, there is at most one index~$i\in \pi_{ew}^{-1}(j^*)$ that also belongs to~$A_2^w$, meaning that the sum in~\eqref{eq:RHS} consists of at most one term. 
We see that sum must is at least~$\beta/2$ and in particular, there is an~$i'\in \pi_{ew}^{-1}(j^*)$ such that~$|b_w(i')| \geq \beta/2$.
We conclude that~$i'\in A_1^w$ and~$\pi_{ew}(i') = j^* = \pi_{ev}(i^*)$, proving the claim.
\end{proof}

\begin{proof}[Proof of \expref{Lemma}{lem-good-labeling}]
By \expref{Proposition}{prop-intersect} and Equation
\eqref{eqn-sizes} any edge $e = (v,w) \in E'$ is satisfied by the assignment~$A$ 
with probability at
least $1/(|A^v_1||A^w_1|) \geq (\eps^4\beta^4)/256$. 
Since by \expref{Proposition}{prop:Eprime}, we have $|E'|\geq (\eps^4/4)|E|$, the
expected fraction of satisfied edges  is at least
$\eps^8\beta^4/1024$.
\end{proof}

\section{The commutative case}
\label{sec:comm}

Recall that the \emph{commutative} Little Grothendieck problem 
asks for the norm of a linear operator $\mathcal F: L_2 \to L_1$.
In this section we use \expref{Theorem}{thm-NP-hard} to prove \expref{Theorem}{thm:comm-hard}, the tight hardness result for this problem.
We first consider the real case of \expref{Theorem}{thm:comm-hard}, and then the complex case in \expref{Section}{sec:comm-complex}.

\subsection{The real case}

The real case of \expref{Theorem}{thm:comm-hard} follows easily by combining \expref{Theorem}{thm-NP-hard} with the following simple lemma.

\begin{lemma}
\label{lem:comm}
For every positive integer~$n$ there exists a map~$f:\R^n \to L_1$ with the following properties:
\begin{itemize}
\item For any vector~$a\in \R^n$, we have~$\length{f(a)}_{L_1} \leq \length{a}_{\ell_2}$.
\item For each standard basis vector~$e_i$, we have~$\length{f(e_i)}_{L_1} = 1$.
If~$\length{f(a)}_{L_1} > (\sqrt{2/\pi}+\eps)\length{a}_{\ell_2}$ then~$\length{a}_{\ell_4} >(\eps/K)\length{a}_{\ell_2}$, where~$K<\infty$ is a universal constant.
\end{itemize}
\end{lemma}

This shows that there is an~$L_1$-valued function~$f$ that satisfies the conditions of the real variant of \expref{Theorem}{thm-NP-hard} for~$\tau = \sqrt{2/\pi}$, $\eta = 1$ and~$\delta(\eps) = (\eps/K)$.
Hence, it is \cclass{NP}-hard to approximate the norm of a linear operator~${\mathcal F}:L_2 \to L_1(L_1)$ over~$\R$ to a factor~$\sqrt{2/\pi} + \eps$ for any~$\eps>0$.
The real case of \expref{Theorem}{thm:comm-hard} then follows from the fact that~$L_1(L_1)$ is isometrically isomorphic to~$L_1$ (\ie, there is a bijective isometry 
between the two).

The proof of \expref{Lemma}{lem:comm} uses the following version of the Berry--Ess\'een Theorem (see for example~\cite[Chapter~5.2, Theorem~5.16]{ODonnell:2014}).

\begin{theorem}[Berry--Ess\'een Theorem]\label{thm:BE}
There exists a universal constant~$K<\infty$ such that the following holds.
Let~$n$ be a positive integer and let~$Z_1,\dots,Z_n$ be independent centered $\pmset{}$-valued random variables.
Then, 
for any $\eps > 0$ and %
for any vector $a\in\R^n$ such that~$\length{a}_{\ell_\infty} \leq\eps\length{a}_{\ell_2}$, we have
\beqn
\Bigg|\Exp\Big[\Big|\sum_{i=1}^n  Z_i a_i \Big|\Big] -\sqrt{\frac{2}{\pi}}\length{a}_{\ell_2}\Bigg|
\leq
K\eps\length{a}_{\ell_2}\,.
\eeqn
\end{theorem}

\begin{proof}[Proof of \expref{Lemma}{lem:comm}]
Endow $\pmset{n}$ with the uniform probability measure and define the function $f:\R^n\to L_1(\pmset{n})$ by
\beqn
\big(f(a)\big)(Z_1,\dots,Z_n) = \sum_{i=1}^na_i Z_i\,.
\eeqn
The first property follows since
\beqn
\length{f(a)}_{L_1}
\leq
\length{f(a)}_{L_2}
=
\left(\Exp\Big[\Big|\sum_{i=1}^na_iZ_i\Big|^2\Big]\right)^{1/2}
=
\length{a}_{\ell_2}\,.
\eeqn
The second property is trivial.
The third property follows from \expref{Theorem}{thm:BE}.
Indeed, the theorem implies that if for some~$\eps>0$, we have
\beqn
\length{f(a)}_{L_1}
=
\Exp\Big[\Big|\sum_{i=1}^na_iZ_i\Big|\Big] 
> 
\left(\sqrt{\frac{2}{\pi}} + \eps\right)\length{a}_{\ell_2}\,,
\eeqn
then~$\length{a}_{\ell_\infty} > (\eps/K)\length{a}_{\ell_2}$.
Since~$\length{a}_{\ell_4} \geq \length{a}_{\ell_\infty}$ the last property follows.
\end{proof}

\subsection{The complex case}
\label{sec:comm-complex}

A similar argument to the one above shows the complex case of \expref{Theorem}{thm:comm-hard}.
This follows from the following complex analogue of \expref{Lemma}{lem:comm}.

\begin{lemma}
\label{lem:comm-complex}
For every positive integer~$n$ there exists a map~$f:\C^n \to L_1$ with the following properties:
\begin{itemize}
\item For any vector~$a\in \C^n$, we have~$\length{f(a)}_{L_1} \leq \length{a}_{\ell_2}$.
\item For each standard basis vector~$e_i$, we have~$\length{f(e_i)}_{L_1} = 1$.
If~$\length{f(a)}_{L_1} > (\sqrt{\pi/4}+\eps)\length{a}_{\ell_2}$ then~$\length{a}_{\ell_4} >(\eps^2/K)\length{a}_{\ell_2}$, where~$K<\infty$ is a universal constant.
\end{itemize}
\end{lemma}

This shows that there is an~$L_1$-valued function~$f$ that satisfies the conditions of \expref{Theorem}{thm-NP-hard} for~$\tau = \sqrt{\pi/4}$, $\eta = 1$ and~$\delta(\eps) = (\eps^2/K)$.
Hence, it is \cclass{NP}-hard to approximate the norm of a linear operator~${\mathcal F}:L_2 \to L_1(L_1)$ over~$\C$ to a factor~$\sqrt{\pi/4} + \eps$ for any~$\eps>0$.
The complex case of \expref{Theorem}{thm:comm-hard} then follows as before from the fact that~$L_1(L_1)$ is isometrically 
isomorphic to~$L_1$.

The proof of \expref{Lemma}{lem:comm-complex} is based on the following complex analogue the Berry--Ess\'een Theorem.
Since we could not find this precise formulation in the literature we include a proof below for completeness.

\begin{lemma}[Complex Berry--Ess\'een Theorem]\label{lem:BEcomplex}
There exists a universal constant~$K<\infty$ such that the following holds.
Let~$Z_1,\dots,Z_n$ be independent uniformly distributed random variables over~$\{1, i, -1, -i\}$.
Then, for any vector~$a\in\C^n$ such that~$\length{a}_{\ell_\infty} \leq \eps\length{a}_{\ell_2}$, we have
\beqn
\Bigg|
\Exp\Big[\Big|\sum_{j=1}^n Z_j a_j\Big|\Big] 
-
\sqrt{\frac{\pi}{4}}\length{a}_{\ell_2}
\Bigg|
\leq 
K\sqrt{\eps}\length{a}_{\ell_2}\,.
\eeqn
\end{lemma}

The proof is based on the following multi-dimensional version of the Berry--Ess\'een theorem due to Bentkus~\cite[Theorem~1.1]{Bentkus:2005}.

\begin{theorem}[Bentkus]\label{thm:bentkus}
Let~$X_1,\dots,X_n$ be independent~$\R^d$-valued random variables such that~$\Exp[X_j] = 0$ for each $j\in[n]$.
Let~$S = X_1 + \cdots + X_n$ and assume that the  covariance matrix of~$S$ equals $\Id_d$.
Let~$g \sim \mathcal N(0, \Id_d)$ be a standard Gaussian vector in~$\R^d$ with the same covariance matrix as~$S$.
Then, for any measurable convex set~$A\subseteq \R^d$, we have
\beqn
\big| \Pr[S\in A] - \Pr[g\in A]\big| \leq c(d)\sum_{j=1}^n\Exp\big[\length{X_j}_{\ell_2}^3\big],
\eeqn
where~$c(d) = O(d^{1/4})$.
\end{theorem}

We also use the following standard tail bounds.

\begin{lemma}[Gaussian tail bound~\cite{Boucheron:2013}]\label{lem:gausstail}
Let~$g\sim\mathcal N(0,\Id_d)$ be a standard Gaussian vector in~$\R^d$.
Then, for any~$t > 0$, we have
\beqn
\Pr\big[\big|\length{g}_{\ell_2} - \sqrt{d}\big| > t\big] 
\leq
2e^{-t^2/2}\,.
\eeqn
\end{lemma}

\begin{lemma}[Hoeffding's inequality~\cite{Hoeffding:1962}]\label{lem:hoeffding}
Let~$X_1,\dots,X_n$ be independent real-valued random variables such that for each~$i\in[n]$, 
$X_i\in[a_i,b_i]$ for some $a_i<b_i$.
Let~$S = X_1 + \cdots + X_n$.
Then, for any~$t>0$, 
\beqn
\Pr\big[|S - \Exp[S]| > t\big] \leq 2e^{-2t^2/\sum_{i=1}^n(b_i - a_i)^2}.
\eeqn
\end{lemma}

\begin{proof}[Proof of \expref{Lemma}{lem:BEcomplex}]
Let~$a\in\C^n$ be some vector. By homogeneity we may assume that~$\length{a}_{\ell_2} = 1$.
Set~$\eps=\length{a}_{\ell_\infty}$.
For each~$j\in[n]$ define the random vector~$X_j \in \R^2$ by $X_j = \sqrt{2}[ \Re(Z_ja_j), \Im(Z_ja_j)]^{\mathsf T}$
and note that $\length{X_j}_{\ell_2} = \sqrt{2}|a_j| \leq \sqrt{2}\eps$.
Let~$S = X_1+\cdots+X_n$,
and let~$T \geq \sqrt{8}$ be some number to be set later.
We have
\beq\label{eq:expYint}
\Exp[\length{S}_{\ell_2}] = \int_0^\infty\Pr[\length{S}_{\ell_2} >t]\,dt
= \int_0^T\Pr[\length{S}_{\ell_2} >t]\,dt + 
   \int_T^\infty\Pr[\length{S}_{\ell_2} >t]\,dt\,.
\eeq

We now analyze each integral separately. 
Notice that~$\Exp[X_j] =0$, and $\Exp[X_jX_j^{\mathsf T}] = |a_j|^2\Id_2$.
It follows that the covariance matrix of~$S$ equals~$\Id_2$.
If we let $g\sim\mathcal N(0,\Id_2)$ be a standard Gaussian vector in~$\R^2$, then it follows from~\expref{Theorem}{thm:bentkus} (for $d=2$) that for any~$t>0$, we have
\begin{align}
\big|\Pr\big[\length{S}_{\ell_2} > t\big] - \Pr\big[\length{g}_{\ell_2} > t\big]\big| 
&\leq 
c\sum_{j=1}^n\Exp\big[\length{X_j}_{\ell_2}^3\big]\nonumber\\
&\leq
\sqrt{2}c\eps \sum_{j=1}^n\Exp\big[\length{X_j}_{\ell_2}^2\big]\nonumber\\
&\leq
2\sqrt{2}c\eps\,.\label{eq:gotze2}
\end{align}
Therefore, the first integral in~\eqref{eq:expYint} satisfies
\begin{align}
&\Big|
\int_0^T\Pr\big[\length{S}_{\ell_2} >t\big]\,dt -
\int_0^T\Pr\big[\length{g}_{\ell_2} >t\big]\,dt
\Big|  \nonumber \\
&\qquad
\le
\int_0^T
\Big|
  \Pr\big[\length{S}_{\ell_2} >t\big] -
  \Pr\big[\length{g}_{\ell_2} >t\big]
\Big|\,dt \nonumber  \\ 
&\qquad
 \le
2\sqrt{2}c\eps T\,. \label{eq:gotzebound}
\end{align}
Since $\length{g}_{\ell_2}$ is distributed according to a~$\chi_2$ distribution, we have 
\beq
\int_0^\infty\Pr\big[\length{g}_{\ell_2} >t \big]\,dt = 
\Exp[\length{g}_{\ell_2}] = 
\sqrt{\pi/2}\,. \label{eq:gaussianexpectation}
\eeq
Moreover, it follows from \expref{Lemma}{lem:gausstail} (for $d = 2$) and our assumption on~$T$ that 
\begin{align}
\int_T^\infty\Pr\big[\length{g}_{\ell_2} -\sqrt{2}>t-\sqrt{2}\big]\,dt 
&\leq
\int_T^\infty2e^{-(t - \sqrt{2})^2/2}\,dt\nonumber\\
&\leq \int_T^\infty 2te^{-t^2/8}\,dt\nonumber\\
&= 8e^{-T^2/8}\,,\label{eq:gausslower}
\end{align}
where we used that $(t - \sqrt{2})^2 \ge t^2/4$ for $t \ge \sqrt{8}$.
Combining~\eqref{eq:gotzebound}, \eqref{eq:gaussianexpectation}, and \eqref{eq:gausslower}, we obtain
that the first integral in~\eqref{eq:expYint} satisfies
\begin{align}
&\Big|
\int_0^T\Pr\big[\length{S}_{\ell_2} >t\big]\,dt -
\sqrt{\pi/2} 
\Big|  
 \le
2\sqrt{2}c\eps T + 8e^{-T^2/8}\,. \label{eq:firstintegralfinal}
\end{align}

We now bound the second integral in~\eqref{eq:expYint}, which is clearly nonnegative. The first coordinate~$S_1$ is a sum of independent random variables, $\sqrt{2}\Re(Z_ja_j)$, which are centered and have magnitude at most~$\sqrt{2}|a_j|$.
Similarly, the same holds for~$S_2$.
\expref{Lemma}{lem:hoeffding} therefore gives,
\begin{align}
\int_T^\infty\Pr[\length{S}_{\ell_2} > t]\,dt
& \leq \int_T^\infty \Pr[|S_1| > t/\sqrt{2}]\,dt + \int_T^\infty\Pr[|S_2| > t/\sqrt{2}]\,dt\nonumber \\
&\leq 4\int_T^\infty e^{-t^2/8}\,dt\nonumber \\
&\leq 4\int_T^\infty te^{-t^2/8}\,dt \nonumber \\
&= 16 e^{-T^2/8}\,,\label{eq:integral}
\end{align}
where in the last inequality we used the assumption~$T\geq 1$.

Now set~$T = \sqrt{8/\eps}$.
Combining~\eqref{eq:expYint}, \eqref{eq:firstintegralfinal}, and~\eqref{eq:integral}, we get
\beqn
\Bigg|\Exp\Big[\Big|\sum_{j=1}^n Z_ja_j\Big|\Big] 
-
\sqrt{\frac{\pi}{4}}\Bigg|
=
\Bigg| \frac{1}{\sqrt{2}}\Exp\big[\length{S}_{\ell_2}\big] 
-
\sqrt{\frac{\pi}{4}}\Bigg|
\leq \frac{1}{\sqrt{2}}\left(2\sqrt{2}c\eps T + 24e^{-T^2/8}\right)
\leq
K\sqrt{\eps}\,.
\qedhere
\eeqn
\end{proof}

The proof of \expref{Lemma}{lem:comm-complex} is nearly identical to that of \expref{Lemma}{lem:comm}, now based on~\expref{Lemma}{lem:BEcomplex} and the function $f:\C^n\to L_1(\{1,i,-1,-i\}^n)$ given by
$
\big(f(a)\big)(Z_1,\dots,Z_n) =a_1Z_1 + \cdots + a_nZ_n,
$
where $\{1,i,-1,-i\}^n$ is endowed with the uniform probability measure.

\section{The non-commutative case}
\label{sec:NC}

In this section we complete the proof of our main theorem (\expref{Theorem}{thm:lncg-ughard}).
The following lemma gives the linear matrix-valued map~$f$ mentioned in the introduction. 

\begin{lemma}\label{lem:ourmap}
Let~$n$ be a positive integer and let~$d = 2^{2n + \ceil{n/2}}$.
Then, there exists a linear operator $f~:~\C^n~\to~\C^{d\times d}$ such that for any vector~$a\in \C^n$, we have
\beqn
\length{f(a)}_{S_1} 
\leq 
\sqrt{\frac{\length{a}_{\ell_2}^2 + \length{a}_{\ell_4}^2}{2}}\,.
\eeqn
In particular,  
$\length{f(a)}_{S_1} \leq (\length{a}_{\ell_2} + \length{a}_{\ell_4})/\sqrt{2}$.
Moreover, for each basis vector~$e_i$ we have~$\length{f(e_i)}_{S_1} = 1$.
\end{lemma}

\expref{Theorem}{thm:lncg-ughard} now follows easily by combining the above lemma with \expref{Theorem}{thm-NP-hard}.
Indeed, \expref{Lemma}{lem:ourmap} shows that the conditions of \expref{Theorem}{thm-NP-hard} hold for~$\tau = 2^{-1/2}$, $\eta = 1$ and $\delta(\eps) = \sqrt{2}\eps$.
It is therefore \cclass{NP}-hard to approximate the norm of a linear operator~$\mathcal F:L_2 \to L_1(S_1)$ to within a factor~$1/\sqrt{2} + \eps$ for any~$\eps >0$.
This implies the theorem because~$L_1(S_1)$ embeds isometrically into~$S_1$. 
To see the last fact, we use the map that takes a matrix-valued function~$g$ on a finite measure space~$U$ to a block diagonal matrix with blocks proportional to $g(u)$ for~$u\in U$ and use the fact that the trace norm of a block diagonal matrix is the average trace norm of the blocks.

The rest of this section is devoted to the proof of~\expref{Lemma}{lem:ourmap}.
For a complex vector~$a\in \C^n$ 
define 
\beq\label{def:parval}
\parea(a) = \sqrt{\length{\Re(a)}_{\ell_2}^2\length{\Im(a)}_{\ell_2}^2 - \big\langle \Re(a),\Im(a)\big\rangle^2}\,.
\eeq
Note that this value is the area of the parallelogram in~$\R^n$ generated by the vectors~$\Re(a)$ and~$\Im(a)$.

\begin{lemma}\label{lem:Tmap}
Let~$n$ be a  positive integer and let~$d' = 2^{\ceil{n/2}}$.
Then, there exists a operator $C:\C^n \to \C^{d'\times d'}$ such that for any vector~$a\in \C^n$, we have
\beq\label{eq:Tmap}
\length{C(a)}_{S_1} = \frac{1}{2}\sqrt{\length{a}_{\ell_2}^2 + 2\parea(a)} + \frac{1}{2}\sqrt{\length{a}_{\ell_2}^2 - 2\parea(a)}\,.
\eeq
\end{lemma}

Though we will not use it here, let us point out that the map~$C$ becomes an isometric embedding if we restrict it to~$\R^n$, since~$\Lambda(a) = 0$ for real vectors.

\begin{proof}
We begin by defining a set of pairwise anti-commuting matrices as follows.
The Pauli matrices are the four Hermitian matrices

\beqn
I=\begin{pmatrix}1&0\\0&1\end{pmatrix},
\quad\quad
X=\begin{pmatrix}0&1\\1&0\end{pmatrix},
\quad\quad
Y=\begin{pmatrix}0&-i\\i&0\end{pmatrix},
\quad\quad
Z=\begin{pmatrix}1&0\\0&-1\end{pmatrix}.
\eeqn

Using these we define $2\ceil{n/2}$ matrices in $\C^{d'\times d'}$ by
\begin{align*}
C_{2j-1} &= \underbrace{Z\otimes \cdots\otimes Z}_{\text{$j-1$ times}} \otimes\, X \otimes \underbrace{I\otimes \cdots\otimes I}_{\text{$\ceil{n/2}-j$ times}}\,,\\[.5cm]
C_{2j} &= \underbrace{Z\otimes \cdots\otimes Z}_{\text{$j-1$ times}} \otimes\, Y \otimes \underbrace{I\otimes \cdots\otimes I}_{\text{$\ceil{n/2}-j$ times}},
\end{align*}
for each~$j \in [\ceil{n/2}]$.
It is easy to verify that these matrices have trace zero, that they are Hermitian, unitary, and that they pairwise anti-commute. 
In particular, they satisfy~$C_j^2 = I$.
For a vector $a\in\C^n$ we define the map~$C$ by $C(a) = a_1C_1 + \cdots + a_nC_n$.
Note that for a \emph{real} vector $x\in\R^n$, the matrix~$C(x)$ is Hermitian and that it satisfies~$C(x)^2 = \|x\|^2_{\ell_2} I$.
If a real vector~$z\in\R^n$ is orthogonal to~$x$ then by expanding the definitions of the matrices~$C(x)$ and~$C(z)$ and using the above properties we find that they anti-commute:
\begin{align*}
C(x)C(z) &= \langle x,z\rangle I + \sum_{j\ne k} x_jz_k C_jC_k\\
&= 0 - \sum_{j\ne k}x_jz_k C_kC_j\\
&= -C(z)C(x)\,.
\end{align*}
This shows that the matrix~$C(x)C(z)$ is skew-Hermitian, which implies that it has purely imaginary eigenvalues.
Since this matrix has trace zero and satisfies $C(x)C(z)\big(C(x)C(z)\big)^* = \|x\|^2_{\ell_2}\|z\|^2_{\ell_2} I$, half the eigenvalues equal~$i\|x\|_{\ell_2}\|z\|_{\ell_2}$ and the other half equal~$-i\|x\|_{\ell_2}\|z\|_{\ell_2}$.

We show that $C$ satisfies~\eqref{eq:Tmap}.
Let~$x = \Re(a)$ and $y = \Im(a)$ so that $C(a) = C(x) + iC(y)$.
Write~$y = y^\parallel + y^\perp$ where $y^\parallel$ is parallel to $x$ and $y^\perp$ is orthogonal to $x$.
Then,
\beqrn
C(a)C(a)^* &=& \big(C(x) + iC(y)\big)\big(C(x) - iC(y)\big)\\
&=& \length{a}_{\ell_2}^2 I - i\big(C(x)C(y) - C(y)C(x)\big)\\
&=& \length{a}_{\ell_2}^2 I - 2i C(x)C(y^\perp)\,,
\eeqrn
where in the last line we used the fact that~$C(y^\parallel)$ commutes with~$C(x)$ while~$C(y^\perp)$ anti-commutes with~$C(x)$.
Using what we deduced above for the matrix~$C(x)C(y^\perp)$ we see that half of the eigenvalues of~$C(a)C(a)^*$ equal $\length{a}_{\ell_2}^2 + 2\length{x}_{\ell_2}\length{y^\perp}_{\ell_2}$ and the other half equal~$\length{a}_{\ell_2}^2 - 2\length{x}_{\ell_2}\length{y^\perp}_{\ell_2}$.
Hence,
\beqn
\length{C(a)}_{S_1} = \frac{1}{2}\sqrt{\length{a}_{\ell_2}^2 + 2\length{x}_{\ell_2}\length{y^\perp}_{\ell_2}} + \frac{1}{2}\sqrt{\length{a}_{\ell_2}^2 - 2\length{x}_{\ell_2}\length{y^\perp}_{\ell_2}}\,.
\eeqn
The claim now follows because~$\length{x}_{\ell_2}\length{y^\perp}_{\ell_2}$ is precisely the area of the parallelogram generated by the vectors~$x$ and~$y$.
\end{proof}

We denote the entry-wise product of two vectors~$a,b\in\C^n$ by~$a\circ b = (a_1b_1,\dots,a_nb_n)$.

\begin{proposition}\label{prop:randphase}
Let $\omega$ be a vector chosen uniformly from $\{1, i, -1, -i\}^n$.
Then, for any~$a\in\C^n$, we have
\beqn
4\Exp_\omega\big[\Lambda(a\circ \omega)^2\big] = \length{a}_{\ell_2}^4 - \length{a}_{\ell_4}^4\,.
\eeqn
\end{proposition}

\begin{proof}
Fix a vector~$a\in\C^n$.
Define the random vectors~$x_\omega = \Re(a\circ \omega)$ and~$y_\omega = \Im(a\circ\omega)$.
Then $\Lambda(a\circ\omega)^2 = \length{x_\omega}_{\ell_2}^2\length{y_\omega}_{\ell_2}^2 - \langle x_\omega,y_\omega\rangle^2$.
For each~$j\in[n]$ we factor~$a_j = \alpha_j e^{i\phi_j}$ and~$\omega_j = e^{i\psi_j}$, where $\alpha_j \in\R_+$ and~$\phi_j,\psi_j\in[0,2\pi]$.
Note that~$\psi_1,\dots,\psi_n$ are independent uniformly distributed random phases in~$\{0,\pi/2,\pi,3\pi/2\}$.
Then,
\beqrn
\length{x_\omega}_{\ell_2}^2 &=& \sum_{j=1}^n \alpha_j^2 \cos^2(\phi_j+\psi_j)\\
\length{y_\omega}_{\ell_2}^2 &=& \sum_{j=1}^n \alpha_j^2 \sin^2(\phi_j+\psi_j)\\
\langle x_\omega,y_\omega\rangle &=& \sum_{j=1}^n \alpha_j^2 \cos(\phi_j+ \psi_j)\sin(\phi_j+\psi_j)\,.
\eeqrn
With this it is easy to verify that
\begin{align}
\length{x_\omega}_{\ell_2}^2\length{y_\omega}_{\ell_2}^2 - \langle x_\omega,y_\omega\rangle^2
=
&\sum_{j\ne k}\alpha_j^2\alpha_k^2\cos^2(\phi_j + \psi_j)\sin^2(\phi_k + \psi_k) - \nonumber\\
& \sum_{j\ne k}\alpha_j^2\alpha_k^2 \cos(\phi_j + \psi_j)\sin(\phi_j + \psi_j)\cos(\phi_k + \psi_k)\sin(\phi_k + \psi_k)\,.
\label{eq:trigexpand}
\end{align}
By independence of~$\psi_j$ and~$\psi_k$ when~$j\ne k$ and the elementary identities 
$
\Exp[\cos^2(\phi_j +\psi_k)] = 1/2$, $\Exp[\sin^2(\phi_j +\psi_k)] = 1/2
$
and
$\Exp[\cos(\phi_j + \psi_j)\sin(\phi_j + \psi_j)]=0$,
the expectation of~\eqref{eq:trigexpand} equals
\beqn
\Exp_\omega\big[\Lambda(a\circ\omega)^2\big] = 
\frac{1}{4}\sum_{j\ne k}a_j^2a_k^2 = \frac{1}{4}\big(\length{a}_{\ell_2}^4 - \length{a}_{\ell_4}^4\big)\,.
\qedhere
\eeqn
\end{proof}

We remark that in the above proof, it suffices if $\omega \in \{1, i, -1, -i\}^n$ is chosen from a pairwise independent family. Using this in the proof below, allows one to prove \expref{Lemma}{lem:ourmap} with a smaller parameter $d$. 

\begin{proof}[Proof of \expref{Lemma}{lem:ourmap}]
Let~$C$ be the map given by \expref{Lemma}{lem:Tmap}.
Define the map
\beqn
f(a) = \bigoplus_{\omega} C(a\circ w)
\eeqn
where~$\omega$ ranges over over~$\{1,i,-1,-i\}^n$.
By convexity of the square function, Jensen's inequality, and the fact that~$\length{a\circ \omega}_{\ell_2} = \length{a}_{\ell_2}$, we have
\begin{align}
\length{f(a)}_{S_1}^2 
&= 
\Big(\Exp_\omega\big[ \length{C(a\circ \omega)}_{S_1} \big]\Big)^2 \nonumber\\
&\stackrel{\text{\expref{Lemma}{lem:Tmap}}}{=} 
\left(\Exp_\omega\left[  \frac{1}{2}\sqrt{\length{a\circ \omega}_{\ell_2}^2 + 2\parea(a\circ\omega)} + \frac{1}{2}\sqrt{\length{a\circ \omega}_{\ell_2}^2 - 2\parea(a\circ\omega)} \right]\right)^2\nonumber\\
&\leq
\Exp_\omega\left[  \left(\frac{1}{2}\sqrt{\length{a}_{\ell_2}^2 + 2\parea(a\circ\omega)} + \frac{1}{2}\sqrt{\length{a}_{\ell_2}^2 - 2\parea(a\circ\omega)}\right)^2 \right]\nonumber\\
&= \frac{\length{a}_{\ell_2}^2}{2} + \frac{1}{2}\Exp_\omega\left[\sqrt{\length{a}_{\ell_2}^4 - 4\Lambda(a\circ \omega)^2}\right].
 \label{eq:ourmap1}
\end{align}
Concavity of the square-root function, Jensen's inequality and \expref{Proposition}{prop:randphase} gives that the expectation in~\eqref{eq:ourmap1} is at most
\beqn
\left(\length{a}_{\ell_2}^4 - 4\Exp\big[\Lambda(a\circ\omega)^2\big] \right)^{1/2}
=
\left(\length{a}_{\ell_2}^4 - \length{a}_{\ell_2}^4 + \length{a}_{\ell_4}^4\big] \right)^{1/2}
=
\length{a}_{\ell_4}^2\,.
\eeqn
Hence,
\beqn
\length{f(a)}_{S_1} \leq \sqrt{\frac{\length{a}_{\ell_2}^2 + \length{a}_{\ell_4}^2}{2}}\,.
\eeqn
For the second claim observe that for any standard basis vector~$e_j$ and~$\omega\in\{1,i,-1,-i\}^n$, the vector~$e_j\circ\omega$ is either purely real or purely imaginary. This implies~$\Lambda(e_j\circ\omega) = 0$.
Hence, by \expref{Lemma}{lem:Tmap},
\beqn
\length{f(e_i)}_{S_1} = 
\frac{1}{2}\Exp_\omega\left[\sqrt{\length{e_j\circ\omega}_{\ell_2}^2 + 2\parea(e_j\circ\omega)} + \sqrt{\length{e_j\circ\omega}_{\ell_2}^2 - 2\parea(e_j\circ\omega)}\right]
= 1\,.
\qedhere
\eeqn
\end{proof}

\subsection{The real and Hermitian variants}
\label{sec:NC-RH}

We end this section by showing that our hardness result of \expref{Theorem}{thm:lncg-ughard} 
also holds for two variants of the Little NCG, the real variant and the Hermitian variant. 
Both variants were introduced (in the context of the ``big'' NCG) in~\cite{Naor:2013}, partly for the purpose of using them in applications. 
The real variant asks for the operator norm of a linear map~$\mathcal F$ from~$\R^n$ to a space $\R^{d\times d}$ endowed with the Schatten-1 norm; in the Hermitian variant, the linear map is from~$\R^n$ to the space $H_d\subseteq\C^{d\times d}$ of Hermitian matrices, again endowed with the Schatten-1 norm. In both cases the operator norm is given by $\length{\mathcal F} = \sup_a\length{F(a)}_{S_1}$ with the supremum over real unit vectors $a$.
Both the real and Hermitian variants follow directly by combining the lemma shown below and the real version of \expref{Theorem}{thm-NP-hard}.
Let us denote by~$\Sym^{d\times d}\subseteq \R^{d\times d}$ the space of real symmetric matrices.

\begin{lemma}\label{lem:RH-map}
Let~$n$ be a positive integer and let~$d$ be as in \expref{Lemma}{lem:ourmap}.
Then, there exists a linear operator $f:\R^n~\to~\Sym^{4d\times 4d}$ satisfying the conditions stated in \expref{Lemma}{lem:ourmap} (with~$a\in\R^n$).
\end{lemma}

The lemma follows by applying the map $\rho$ of the elementary claim below
to the restriction of the operator~$f$ of \expref{Lemma}{lem:ourmap} to $\R^n$.

\begin{claim}
For every positive integer~$d$ there exists a map~$\rho:\C^{d\times d} \to \Sym^{4d\times 4d}$ such that for any matrix $A\in\C^{d\times d}$, we have~$\length{\rho(A)}_{S_1} = \length{A}_{S_1}$. Moreover,~$\rho$ is linear over the real numbers, that is, for any~$\alpha\in \R$ and~$A,B\in\C^{d\times d}$, we have~$\rho(\alpha A) = \alpha\rho(A)$ and~$\rho(A+B) = \rho(A) + \rho(B)$.
\end{claim}

\begin{proof}
The proof follows by combining two standard transformations taking complex matrices to Hermitian matrices and real matrices, respectively.
Let~$A\in\C^{d\times d}$ be a matrix with singular values $\sigma_1\geq\cdots\geq \sigma_d$.
The first transformation is given by~$A\mapsto \big[\begin{smallmatrix}0& A\\A^*& 0\end{smallmatrix}\big]$.
By~\cite[Theorem~7.3.3]{HornJohnson}, the last matrix has eigenvalues~$\sigma_1 \geq\cdots\geq \sigma_d\geq -\sigma_d\geq \cdots\geq -\sigma_1$.
Notice that this transformation is linear over the reals since the adjoint is such.
Let~$B\in\C^{d\times d}$ be a Hermitian matrix with eigenvalues $\lambda_1\geq \cdots\geq\lambda_d$. 
The second transformation is given by
\[
  B\mapsto \begin{bmatrix}\Re(B) &\Im(B)\\-\Im(B) &\Re(B)\end{bmatrix}\,.
\]
Then the last matrix is symmetric and by~\cite[1.30.P20~(g), p.~71]{HornJohnson}, that matrix has the same eigenvalues as~$B$  but with doubled multiplicities, that is, the matrix has eigenvalues $\lambda_1\geq\lambda_1\geq \cdots\geq \lambda_d\geq \lambda_d$.
Notice that this transformation is also linear over the reals.
Let~$\rho$ be the composition of these maps.
Then the matrix~$\rho(A)$ has the same singular values as~$A$ but with quadrupled multiplicities, which implies that~$\length{\rho(A)}_{S_1} = \length{A}_{S_1}$, and~$\rho$ is linear over the reals.
\end{proof}

\section{Little versus big Grothendieck theorem}
\label{sec:little-big}

For completeness, we include here the well-known relation between the little and big Grothendieck problems. 
We focus on the non-commutative case; the commutative case is similar and can be found in, \eg,~\cite[Section~5]{Pisier:2012}.
This discussion clarifies how to derive~\expref{Theorem}{thm:ncgt-ughard} from \expref{Theorem}{thm:lncg-ughard}.

Consider a linear map~$\mathcal F: \C^n \to S_1^d$.
A standard and easy-to-prove fact is that for two finite-dimensional Banach spaces~$X,Y$, the operator norm of a linear map~$\mathcal G: X\to Y$ equals the norm of its adjoint~$\mathcal G^*: Y^* \to X^*$.
As a result, $\length{\mathcal F}  = \length{\mathcal F^*}$.
Notice that since Hilbert space is self-dual and the dual of~$S_1$ is the space~$S_\infty$ of matrices endowed with the Schatten-$\infty$ norm (\ie, the maximum singular value), 
we have that $\mathcal F^*: S_\infty^d \to \C^n$. 
In particular, 
\[
\length{\mathcal F^*} = \sup \length{\mathcal F^*(A)}_2\,,
\]
where the supremum is taken over all $A$ of Schatten-$\infty$ norm at most $1$.
Equivalently, since any matrix with Schatten-$\infty$ norm at most~1 lies in the convex hull of the set of unitary matrices, we could take the supremum over all unitary matrices $A$.
Next, recall that in the NCG problem we 
are given a bilinear form~$T: \C^{d\times d}\times\C^{d\times d}\to \C$, 
and asked to compute $\opt(T) = \sup_{A,B} \big|T(A,B)\big|$, where the supremum ranges over unitary matrices.
Define the bilinear form $T(A,B) = \langle \mathcal F^*(A), \mathcal F^*(B)\rangle$.
By Cauchy-Schwarz,
\beqn
\opt(T) = \sup_A\length{\mathcal F^*(A)}_2^2 =
\length{\mathcal F^*}^2  = \length{\mathcal F}^2\,,
\eeqn
where the supremum is over all unitary~$A$,
showing that the Little NCG is a special case of the ``big'' NCG.

\bibliographystyle{tocplain}   %
\bibliography{v013a015}

\providecommand{\bibhead}[1]{}
\expandafter\ifx\csname pdfbookmark\endcsname\relax%
  \providecommand{\tocrefpdfbookmark}{}
\else\providecommand{\tocrefpdfbookmark}{%
   \hypertarget{tocreferences}{}%
   \pdfbookmark[1]{References}{tocreferences}}%
\fi

\tocrefpdfbookmark
\begin{thebibliography}{10}

\bibitem{Alon:2006}\bibhead{Alon:2006}
{\sc Noga Alon and Assaf Naor}: Approximating the cut-norm via {G}rothendieck's
  inequality.
\newblock {\em SIAM J. Comput.}, 35(4):787--803, 2006.
\newblock Preliminary version in
  \href{http://dx.doi.org/10.1145/1007352.1007371}{STOC'04}.
\newblock [\epfmtdoi{10.1137/S0097539704441629}]

\bibitem{Bandeira:2013}\bibhead{Bandeira:2013}
{\sc Afonso~S. Bandeira, Christopher Kennedy, and Amit Singer}: Approximating
  the little {G}rothendieck problem over the orthogonal group.
\newblock {\em Math. Program.}, 160(1-2):433--475, 2016.
\newblock [\epfmtdoi{10.1007/s10107-016-0993-7}, \epfmt{arxiv}{1308.5207}]

\bibitem{BarakBHKSZ12}\bibhead{BarakBHKSZ12}
{\sc Boaz Barak, Fernando G. S.~L. Brand{\~{a}}o, Aram~W. Harrow, Jonathan~A.
  Kelner, David Steurer, and Yuan Zhou}: Hypercontractivity, sum-of-squares
  proofs, and their applications.
\newblock In {\em Proc. 44th STOC}, pp. 307--326. ACM Press, {2012}.
\newblock [\epfmtdoi{10.1145/2213977.2214006}, \epfmt{arxiv}{1205.4484}]

\bibitem{Bentkus:2005}\bibhead{Bentkus:2005}
{\sc Vidmantas Bentkus}: A {Lyapunov} type bound in $\mathbb{R}^d$.
\newblock {\em Theory Probab. Appl.}, 49(2):311--323, 2005.
\newblock Translated from Russian.
\newblock [\epfmtdoi{10.1137/S0040585X97981123}]

\bibitem{Boucheron:2013}\bibhead{Boucheron:2013}
{\sc St{\'e}phane Boucheron, G{\'a}bor Lugosi, and Pascal Massart}: {\em
  Concentration Inequalities}.
\newblock Oxford Univ. Press, 2013.

\bibitem{Braverman:2013}\bibhead{Braverman:2013}
{\sc Mark Braverman, Konstantin Makarychev, Yury Makarychev, and Assaf Naor}:
  The {G}rothendieck constant is strictly smaller than {K}rivine's bound.
\newblock {\em Forum Math. Pi}, 1:e4, 42, 2013.
\newblock Preliminary version in
  \href{http://dx.doi.org/10.1109/FOCS.2011.77}{FOCS'11}.
\newblock [\epfmtdoi{10.1017/fmp.2013.4}]

\bibitem{BrietFV10}\bibhead{BrietFV10}
{\sc Jop Bri{\"e}t, Fernando~M{\'a}rio de~Oliveira~Filho, and Frank Vallentin}:
  The {G}rothendieck problem with rank constraint.
\newblock In {\em Proc. 19th Symp. Mathem. Theory of Networks and Systems
  ({MTNS}'10)}, pp. 111--113, 2010.
\newblock
  \href{http://www.conferences.hu/mtns2010/proceedings/Papers/018_335.pdf}{MTNS}.

\bibitem{BrietRegevSaket:FOCS}\bibhead{BrietRegevSaket:FOCS}
{\sc Jop Bri{\"{e}}t, Oded Regev, and Rishi Saket}: Tight hardness of the
  non-commutative {G}rothendieck problem.
\newblock In {\em Proc. 56th FOCS}, pp. 1108--1122. IEEE Comp. Soc. Press,
  {2015}.
\newblock [\epfmtdoi{10.1109/FOCS.2015.72}]

\bibitem{Davie:1984}\bibhead{Davie:1984}
{\sc Alexander Davie}: Lower bound for {$K_G$}.
\newblock Unpublished, 1984.

\bibitem{FeigeS02}\bibhead{FeigeS02}
{\sc Uriel Feige and Gideon Schechtman}: On the optimality of the random
  hyperplane rounding technique for {MAX} {CUT}.
\newblock {\em Random Structures Algorithms}, 20(3):403--440, 2002.
\newblock [\epfmtdoi{10.1002/rsa.10036}]

\bibitem{Grothendieck:1953}\bibhead{Grothendieck:1953}
{\sc Alexander Grothendieck}: {R\'esum\'e de la th\'eorie m\'etrique des
  produits tensoriels topologiques (French)}.
\newblock {\em Bol.\ Soc.\ Mat.\ S\~{a}o Paulo}, 8:1--79, 1953.
\newblock Available from
  \href{https://www.ime.usp.br/acervovirtual/textos/estrangeiros/grothendieck/produits_tensoriels_topologiques/files/produits_tensoriels_topologiques.pdf}{Instituto
  de Matem{\'a}tica e Estat{\'i}stica da Universidade de S{\~a}o Paulo}.

\bibitem{GRSW}\bibhead{GRSW}
{\sc Venkatesan Guruswami, Prasad Raghavendra, Rishi Saket, and Yi~Wu}:
  Bypassing {UGC} from some optimal geometric inapproximability results.
\newblock {\em ACM Trans. Algor.}, 12(1):6:1--6:25, 2016.
\newblock Preliminary version in
  \href{http://dl.acm.org/citation.cfm?id=2095116.2095174}{SODA'12}.
\newblock [\epfmtdoi{10.1145/2737729}]

\bibitem{Haagerup:1985}\bibhead{Haagerup:1985}
{\sc Uffe Haagerup}: The {G}rothendieck inequality for bilinear forms on
  ${C}^*$-algebras.
\newblock {\em Adv. in Math.}, 56(2):93 -- 116, 1985.
\newblock [\epfmtdoi{10.1016/0001-8708(85)90026-X}]

\bibitem{Haagerup:1987}\bibhead{Haagerup:1987}
{\sc Uffe Haagerup}: A new upper bound for the complex {G}rothendieck constant.
\newblock {\em Israel J. Math.}, 60(2):199--224, 1987.
\newblock [\epfmtdoi{10.1007/BF02790792}]

\bibitem{Haagerup:1995}\bibhead{Haagerup:1995}
{\sc Uffe Haagerup and Takashi Itoh}: Grothendieck type norms for bilinear
  forms on {$C^*$}-algebras.
\newblock {\em J. Operator Theory}, 34(2):263--283, 1995.
\newblock Available from \href{http://www.jstor.org/stable/24714900}{JSTOR}.

\bibitem{Hastad:2001}\bibhead{Hastad:2001}
{\sc Johan H{\aa}stad}: Some optimal inapproximability results.
\newblock {\em J. ACM}, 48(4):798--859, 2001.
\newblock Preliminary version in
  \href{http://dx.doi.org/10.1145/258533.258536}{STOC'97}.
\newblock [\epfmtdoi{10.1145/502090.502098}]

\bibitem{Hoeffding:1962}\bibhead{Hoeffding:1962}
{\sc Wassily Hoeffding}: Probability inequalities for sums of bounded random
  variables.
\newblock {\em J. Amer. Statist. Assoc.}, 58(301):13--30, 1963.
\newblock [\epfmtdoi{10.1007/978-1-4612-0865-5\_26}]

\bibitem{HornJohnson}\bibhead{HornJohnson}
{\sc Roger~A. Horn and Charles~R. Johnson}: {\em Matrix Analysis}.
\newblock Cambridge Univ. Press, 1990.

\bibitem{Khot02-color}\bibhead{Khot02-color}
{\sc Subhash Khot}: Hardness results for coloring 3-colorable 3-uniform
  hypergraphs.
\newblock In {\em Proc. 43rd FOCS}, pp. 23--32. IEEE Comp. Soc. Press, {2002}.
\newblock [\epfmtdoi{10.1109/SFCS.2002.1181879}]

\bibitem{Khot:CCC2010}\bibhead{Khot:CCC2010}
{\sc Subhash Khot}: On the unique games conjecture.
\newblock In {\em Proc. 25th IEEE Conf. on Computational Complexity (CCC'10)},
  pp. 99--121. IEEE Comp. Soc. Press, {2010}.
\newblock [\epfmtdoi{10.1109/CCC.2010.19}]

\bibitem{KKMO}\bibhead{KKMO}
{\sc Subhash Khot, Guy Kindler, Elchanan Mossel, and Ryan O'Donnell}: Optimal
  inapproximability results for {MAX}-{CUT} and other 2-variable {CSP}s?
\newblock {\em SIAM J. Comput.}, 37(1):319--357, 2007.
\newblock Preliminary version in
  \href{http://dx.doi.org/10.1109/FOCS.2004.49}{FOCS'04}.
\newblock [\epfmtdoi{10.1137/S0097539705447372}]

\bibitem{Khot:2009}\bibhead{Khot:2009}
{\sc Subhash Khot and Assaf Naor}: Approximate kernel clustering.
\newblock {\em Mathematika}, 55(1-2):129--165, 2009.
\newblock Preliminary version in
  \href{http://dx.doi.org/10.1109/FOCS.2008.33}{FOCS'08}.
\newblock [\epfmtdoi{10.1112/S002557930000098X}, \epfmt{arxiv}{0807.4626}]

\bibitem{Khot:2012}\bibhead{Khot:2012}
{\sc Subhash Khot and Assaf Naor}: Grothendieck-type inequalities in
  combinatorial optimization.
\newblock {\em Comm. Pure Appl. Math.}, 65(7):992--1035, 2012.
\newblock [\epfmtdoi{10.1002/cpa.21398}, \epfmt{arxiv}{1108.2464}]

\bibitem{Khot:maxcutgain}\bibhead{Khot:maxcutgain}
{\sc Subhash Khot and Ryan O'Donnell}: {SDP} gaps and {UGC}-hardness for
  {Max}-{Cut}-{Gain}.
\newblock {\em Theory of Computing}, 5(4):83--117, 2009.
\newblock Preliminary version in
  \href{http://dx.doi.org/10.1109/FOCS.2006.67}{FOCS'06}.
\newblock [\epfmtdoi{10.4086/toc.2009.v005a004}]

\bibitem{Kindler:2010}\bibhead{Kindler:2010}
{\sc Guy Kindler, Assaf Naor, and Gideon Schechtman}: The {UGC} hardness
  threshold of the {$L_p$} {G}rothendieck problem.
\newblock {\em Math. Oper. Res.}, 35(2):267--283, 2010.
\newblock Preliminary version in
  \href{https://dl.acm.org/citation.cfm?id=1347082.1347090}{SODA'08}.
\newblock [\epfmtdoi{10.1287/moor.1090.0425}]

\bibitem{Krivine:79a}\bibhead{Krivine:79a}
{\sc Jean-Louis Krivine}: Constantes de {G}rothendieck et fonctions de type
  positif sur les sph{\`e}res.
\newblock {\em Adv. Math.}, 31(1):16--30, 1979.
\newblock [\epfmtdoi{10.1016/0001-8708(79)90017-3}]

\bibitem{Lindenstrauss:1968}\bibhead{Lindenstrauss:1968}
{\sc Joram Lindenstrauss and Aleksander Pe{\l}czy{\'n}ski}: Absolutely summing
  operators in {$L\sb{p}$}-spaces and their applications.
\newblock {\em Studia Math.}, 29(3):275--326, 1968.
\newblock Available from \href{https://eudml.org/doc/217232}{EuDML}.

\bibitem{Naor:2013}\bibhead{Naor:2013}
{\sc Assaf Naor, Oded Regev, and Thomas Vidick}: Efficient rounding for the
  noncommutative {G}rothendieck inequality.
\newblock {\em Theory of Computing}, 10(11):257--295, 2014.
\newblock Preliminary version in
  \href{http://dx.doi.org/10.1145/2488608.2488618}{STOC'13}.
\newblock [\epfmtdoi{10.4086/toc.2014.v010a011}]

\bibitem{Nesterov:1998}\bibhead{Nesterov:1998}
{\sc Yurii Nesterov}: Semidefinite relaxation and nonconvex quadratic
  optimization.
\newblock {\em Optim. Methods Softw.}, 9(1-3):141--160, 1998.
\newblock [\epfmtdoi{10.1080/10556789808805690}]

\bibitem{ODonnell:2014}\bibhead{ODonnell:2014}
{\sc Ryan O'Donnell}: {\em Analysis of Boolean Functions}.
\newblock Cambridge Univ. Press, 2014.
\newblock Available at
  \href{https://pdfs.semanticscholar.org/51b5/05218ada899916c45f4cd04a7e8df19265f4.pdf}{Semantic
  Scholar}.
\newblock [\epfmtdoi{10.1017/CBO9781139814782}]

\bibitem{Pisier:1978}\bibhead{Pisier:1978}
{\sc Gilles Pisier}: Grothendieck's theorem for noncommutative
  {$C^{\ast}$}-algebras, with an appendix on {G}rothendieck's constants.
\newblock {\em J. Funct. Anal.}, 29(3):397--415, 1978.
\newblock [\epfmtdoi{10.1016/0022-1236(78)90038-1}]

\bibitem{Pisier:2012}\bibhead{Pisier:2012}
{\sc Gilles Pisier}: Grothendieck's theorem, past and present.
\newblock {\em Bull. Amer. Math. Soc.}, 49(2):237--323, 2012.
\newblock [\epfmtdoi{10.1090/S0273-0979-2011-01348-9},
  \epfmt{arxiv}{1101.4195}]

\bibitem{Raghavendra:2009}\bibhead{Raghavendra:2009}
{\sc Prasad Raghavendra and David Steurer}: Towards computing the
  {G}rothendieck constant.
\newblock In {\em Proc. 20th ACM-SIAM Symp. on Discrete Algorithms (SODA'09)},
  pp. 525--534. ACM Press, {2009}.
\newblock \href{https://dl.acm.org/citation.cfm?id=1496770.1496828}{ACM DL}.

\bibitem{RegevV12a}\bibhead{RegevV12a}
{\sc Oded Regev and Thomas Vidick}: Quantum {XOR} games.
\newblock {\em ACM Trans. Comput. Theory}, 7(4):15:1--15:43, 2015.
\newblock Preliminary version in
  \href{http://dx.doi.org/10.1109/CCC.2013.23}{CCC'13}.
\newblock [\epfmtdoi{10.1145/2799560}, \epfmt{arxiv}{1207.4939}]

\bibitem{Rietz:1974}\bibhead{Rietz:1974}
{\sc Ronald~E. Rietz}: A proof of the {G}rothendieck inequality.
\newblock {\em Israel J. Math.}, 19(3):271--276, 1974.
\newblock [\epfmtdoi{10.1007/BF02757725}]

\bibitem{Trevisan:2012}\bibhead{Trevisan:2012}
{\sc Luca Trevisan}: On {K}hot's unique games conjecture.
\newblock {\em Bull. Amer. Math. Soc. (N.S.)}, 49(1):91--111, 2012.
\newblock [\epfmtdoi{10.1090/S0273-0979-2011-01361-1}]

\bibitem{Tsirelson:85b}\bibhead{Tsirelson:85b}
{\sc Boris~S. Tsirel'son}: Quantum analogues of the \uppercase{B}ell
  inequalities. \uppercase{T}he case of two spatially separated domains.
\newblock {\em J. Soviet Math.}, 36(4):557--570, 1987.
\newblock [\epfmtdoi{10.1007/BF01663472}]

\end{thebibliography}

\begin{tocauthors}
\begin{tocinfo}[briet]
Jop Bri\"{e}t\\
Assistant professor\\
 CWI, Amsterdam\\
 The Netherlands\\
 j.briet\tocat{}cwi\tocdot{}nl \\
 \url{http://homepages.cwi.nl/~jop/}
\end{tocinfo}
\begin{tocinfo}[regev]
  Oded Regev\\
  Professor\\
  Courant Institute of Mathematical Sciences\\
  New York University\\
  New York, N.Y.\\
  regev\tocat{}cims\tocdot nyu\tocdot edu \\
  \url{http://www.cims.nyu.edu/~regev/}
\end{tocinfo}
\begin{tocinfo}[saket]
  Rishi Saket\\
  Researcher \\
  IBM Research\\
  Bangalore, India\\
  rissaket\tocat{}in\tocdot ibm\tocdot com \\
  \url{http://researcher.ibm.com/researcher/view.php?person=in-rissaket}
\end{tocinfo}
\end{tocauthors}

\begin{tocaboutauthors}
\begin{tocabout}[briet]
\textsc{Jop Bri\"{e}t} graduated from \href{http://www.cwi.nl}{CWI}
in 2011; his advisor was 
\href{http://homepages.cwi.nl/~buhrman/}{Harry Buhrman}.
He enjoys problems at the intersection of theoretical computer science and 
pure mathematics.
After a stint as a yogi during a postdoc at the 
\href{https://www.cims.nyu.edu/}{Courant Institute} in New York City, 
he returned to rock climbing in the flattest country in the world.
\end{tocabout}
\begin{tocabout}[regev]
\textsc{Oded Regev} graduated from
  \href{http://www.tau.ac.il/}{Tel Aviv University} in 2001 under the
  supervision of \href{http://www.cs.tau.ac.il/~azar/}{Yossi Azar}.
  He spent two years as a postdoc at the 
\href{http://www.ias.edu/}{Institute for Advanced Study}, Princeton, 
and one year at the
  \href{http://www.berkeley.edu/}{University of California, Berkeley}.
  He is currently with the Courant Institute of Mathematical Sciences, and 
enjoys life in NYC\@.
  His research interests include computational and mathematical aspects of 
lattices, quantum computation, and other topics in theoretical computer
  science. 
\end{tocabout}
\begin{tocabout}[saket]
\textsc{Rishi Saket} completed his \phd\ from 
\href{http://www.gatech.edu/}{Georgia Tech} in 2009 under the supervision of
\href{http://www.cs.nyu.edu/~khot/}{Subhash Khot}. After
post-doctoral stints at CMU, Princeton University, and IBM T.\,J.
Watson, he joined IBM Research, Bangalore, India in 2013 where he is
currently a researcher. His interests are in hardness of approximation, 
approximation algorithms, optimization, learning theory, 
and related areas of theoretical and
applied computer sciences.
\end{tocabout}
\end{tocaboutauthors}

\end{document}